\documentclass{tran-l}
\usepackage{amsmath, amssymb}
\usepackage{graphicx}
\usepackage{color}
\usepackage{textcomp}
\usepackage{graphics}

 \newtheorem{thm}{Theorem}[section]
 \newtheorem{cor}[thm]{Corollary}
 \newtheorem{lem}[thm]{Lemma}
 \newtheorem{prop}[thm]{Proposition}
 \theoremstyle{definition}
 
 \theoremstyle{remark}
 \newtheorem{rem}[thm]{Remark}
 \newtheorem{exa}[thm]{Example}
 \newcommand{\PP}{\mathbb{P}}

\newcommand{\bm}{\bibitem}
\newcommand{\no}{\noindent}
\newcommand{\be}{\begin{equation}}
\newcommand{\ee}{\end{equation}}

\newcommand{\bea}{\begin{eqnarray}}
\newcommand{\bes}{\begin{subequations}}
\newcommand{\ees}{\end{subequations}}
\newcommand{\bgt}{\begin{gather}}
\newcommand{\egt}{\begin{gather}}
\newcommand{\eea}{\end{eqnarray}}
\newcommand{\beaa}{\begin{eqnarray*}}
\newcommand{\eeaa}{\end{eqnarray*}}

\newcommand{\EE}{{\mathbb E}}
\newcommand{\RR}{{\mathbb R}}
\newcommand{\cal}{\mathcal}

\newcommand{\FF}{{\mathcal F}}

\begin{document}
\title[Martingale Schr\"{o}dinger bridges]{From (Martingale) Schr\"{o}dinger bridges to \\
a new class of Stochastic Volatility Models}
\author{Pierre Henry-Labord\`ere}
\address{Soci\'et\'e G\'en\'erale, Global markets Quantitative Research}
\address{CMAP, Ecole Polytechnique}
\email{pierre.henry-labordere@sgcib.com} \subjclass{} \keywords{}

\maketitle
\begin{abstract} Following closely the construction of the Schr\"{o}dinger bridge, we build a new class of Stochastic Volatility Models exactly calibrated to market instruments such as for example Vanillas, options on realized variance or VIX options. These models differ strongly from the well-known local stochastic volatility models, in particular the instantaneous volatility-of-volatility of the associated naked SVMs is not modified, once calibrated to market instruments. They can be interpreted as a martingale version of the Schr\"{o}dinger bridge. The numerical calibration is  performed  using a dynamic-like version of the  Sinkhorn algorithm. We finally highlight a striking relation with Dyson non-colliding Brownian motions.
\end{abstract}

\section{Introduction}

\subsection{Motivation: a new class of SVMs}
\no Let us consider a stochastic volatility model (in short SVM\footnote{Not to be confused with Support Vector Machine!}) defined under a (risk-neutral)-measure $\PP^0$ by a time-homogenous It\^o diffusion:
\bea dS_t&=&S_t a_t dW^0_t \label{nakedSVM}\\
da_t&=&b(a_t)dt+ \sigma(a_t) dZ^0_t,  \quad d\langle Z^0, W^0\rangle_t=\rho dt \nonumber \eea under which $S_t$ is a (local) $\PP^0$-martingale (true martingale under proper assumptions on the coefficients $b(\cdot)$ and $\sigma(\cdot)$ and the correlation $\rho$). Here $ W^0_t$ and $Z^0_t$ are two correlated $\PP^0$-Brownian motions. As well-known examples, one can cite the Heston model, the SABR model and the Bergomi model (see \cite{ber} for an exhaustive list of examples). In the present paper,  we consider mainly one-factor SVMs  although the extension to the multi-dimensional setup  will be discussed. Moreover, for the sake of simplicity of notations, we have assume a zero rate. This can be trivially extended to a deterministic rate by considering the forward process. In practice, the volatility-drift  $b$ and the volatility-of-volatility $\sigma$ depend on some parameters (vol-of-vol, mean-reversion, ...) in addition to the spot/volatility correlation $\rho$. As depending on a finite number of parameters, this model is not perfectly calibrated (at $t=0$) to the market values $C_\mathrm{mkt}(T,K)$ of call options, with payoff $(S_T-K)^+$, for all maturities $T$ and for all strikes $K$, meaning that for all $(T,K) \in (0,T_{\max}] \times \RR_+$ almost everywhere:
\beaa \EE^{\PP^0}[(S_T-K)^+] \neq C_\mathrm{mkt}(T,K) \eeaa
The calibration to Vanillas can however be  achieved by {\it modifying} the dynamics of our SVM (under a measure $\PP$) into
\bea dS_t&=&\sigma(t,S_t) a_t dW_t  \label{LSVM} \\
da_t&=&b(a_t)dt+ \sigma(a_t) dZ_t,  \quad d\langle Z, W\rangle_t=\rho dt \nonumber \eea where we have added  a deterministic function $\sigma(t,S_t)$ of the time $t$ and the spot price $S_t$ on top of  the volatility $a_t$. This extension corresponds to the so-called local SVMs, first introduced in \cite{lip}. By a straightforward application of It\^o-Tanaka's lemma on the convex payoff $(S_t-K)^+$, one can then show (see Chapter 11  in \cite{phl1} for a detailed derivation) that this model is calibrated to Vanillas, i.e., $\EE^{\PP}[(S_T-K)^+] = C_\mathrm{mkt}(T,K)$ for all $(T,K) \in \RR_+^2$, if and only if
\bea  \sigma(t,S_t)^2 \EE^\PP[ a_t^2|S_t] = \sigma_\mathrm{loc}(t,S_t)^2 \label{universal} \eea  where $\sigma_\mathrm{loc}(t,K):=2{\partial_t C_\mathrm{mkt}(t,K) \over \partial_K^2 C_\mathrm{mkt}(t,K)}$ is the Dupire local volatility \cite{dup}. By injecting $\sigma$ from equation (\ref{universal}) into SDE (\ref{LSVM}), this leads to a non-linear McKean SDE:
\beaa dS_t&=&{\sigma_\mathrm{loc}(t,S_t) \over \sqrt{\EE^\PP[a_t^2|S_t]}}a_t dW_t  \eeaa The numerical simulation of such a nonlinear SDE can then be achieved efficiently using a particle method (or a PDE numerical scheme for the associated nonlinear Fokker-Planck PDE -- see \cite{phl1} for extensive details and references) and this is one of the main(/only) reason why this modification with a multiplicative function $\sigma$ has been considered by practitioners in mathematical finance. Despite this numerical efficiency, let us remark that in order to fit Vanillas (or equivalently prescribe marginals), we have  drastically modified the dynamics of the resulting instantaneous volatility $A_t:=\sigma(t,S_t) a_t/S_t$ which is now given under $\PP$ by ($\PP$ and $\PP^0$ are not equivalent  probability measures):
\beaa dS_t&=&S_t A_t dW_t \\
{dA_t \over A_t}&=&{ da_t \over a_t} +S^2_t \partial_S \left( \ln {\sigma(t,S_t) \over S_t}\right) A_t dW_t
+(\cdots)dt \eeaa  See the additional  term $S_t  \partial_S \left( \ln {\sigma(t,S_t) / S_t}\right) A_t dW_t$ in the diffusion of $A_t$. In the following paper, we explain how to {\it slightly} deform our naked SVM (as defined by SDE (\ref{nakedSVM})) in order to fit  Vanillas. This deformation consists in adding a drift $\lambda(t,S_t,a_t)$ to the volatility process, without modifying the  volatility-of-volatility as in LSVMs, in particular $\PP$ is equivalent  to $\PP^0$ here. The dynamics of our calibrated SVM reads now:
\beaa dS_t&=&S_t a_t dW_t \\
da_t&=&(b(a_t)+\lambda(t,S_t,a_t))dt+ \sigma(a_t) dZ_t,  \quad d\langle Z, W\rangle_t=\rho dt \eeaa \no \no Note that modeling the correct volatility-of-volatility $\sigma(a_t)$ is still a relevant subject and it is not considered in the present paper (see \cite{ber} for extensive discussions --  a relevance choice is for example to take a log-normal diffusion $\sigma(a)=\nu a$ as in the SABR or (one-factor) Bergomi model).

\no Our approach follows closely the construction of the so-called Schr\"{o}dinger bridge using an entropy penalty. A similar approach was explored in \cite{ave,ave1}. The martingality constraints seem however to have been unnoticed, resulting in pricing models that are not arbitrage-free.  This is confirmed in proposition 4 in \cite{ave1}, where the drift of the diffusion measure $\PP$ is computed and found to be different from the risk-free interest rate.

\subsection{Contents}
The contents of our paper is as follows: In the first section, as a toy model, we recall the construction of the Schr\"{o}dinger bridge \cite{leo}.  We then explicit the link with the theory of conditioned SDEs as developed in \cite{bau}.  In particular, we consider conditioned SDEs to have multiple fixed marginals at maturities $(t_i)_{1 \leq i\leq n}$. We then move on to mathematical finance and explain how to deform a  SVM in order to match some marginals (i.e., Vanilla options). From a mathematical point of view, the  Schr\"{o}dinger bridge  is now restricted to be a martingale and has fixed marginals.  The numerical algorithm for computing the drift $\lambda$ boils down to  the solution of a low-dimensional concave optimization with a (martingale) Sinkhorn algorithm. In the third section, we extend our construction  and explain how to calibrate path-dependent options. As a striking example, we consider SVM calibrated to options on variance depending on the quadratic variation $\langle S \rangle_T$ at some maturity $T$. Finally, in the last section, we highlight a striking relation with the theory of non-colliding diffusions, in particular reproduce the Dyson Brownian motion, related to GOE ensemble in random matrix theory.

\section{An appetizer: Schr\"{o}dinger bridges}

\no Let us consider  a standard $\PP^0$-Brownian motion:
\beaa dX_t= dW^0_t, \quad X_{t=0}:=X_0 \eeaa
\no In this section, as an appetizer,  we consider the problem of  adding a drift to $X_t$ such that the law of the new resulting process $\bar{X}_T$ at a maturity $T$ matches a marginal distribution $\mu$, i.e.,  $\bar{X}_T \sim \mu$.

\begin{rem}[mapping]
Note that if we use a mapping $\bar{X}_t:=f(W^0_t)$ where $f$ is a monotone function chosen such that $\bar{X}_T \sim \mu$ ($F_\mu$ is the cumulative distribution of $\mu$), i.e.,
\beaa F_\mu(x)=\EE^{\PP^0}[1_{W^0_T<f^{-1}(x)}]:={\cal N}(f^{-1}(x)) \Longleftrightarrow f(x):=F_\mu^{-1}\circ{\cal N}(x) \eeaa the volatility of $X_t$ (and its drift) will be modified according to:
\beaa d\bar{X}_t=\partial_x f(W^0_t) dW^0_t + {1 \over 2}\partial_{x}^2f(W^0_t) dt \eeaa This mapping can be seemed as an analog to our modification with local SVMs, see the modification of the diffusion of $W_t^0$.
\end{rem}
\no For use below, $W_t^0$ denotes a standard Brownian w.r.t to a probability measure       $\PP^0$.
\subsection{One marginal}
\begin{prop}[Schr\"{o}dinger, one marginal \cite{sch}] \label{thm1} Let us consider the static strictly concave optimization:
\bea P_1:=\sup_{f \in \mathrm{L}^1(\mu)} \{ -\EE^{\mu}[f] -\ln \EE^{\PP^0}[ e^{-f(W^0_T)}|W^0_0=X_0] \}\label{optP1} \eea and assume that $P_1<\infty$. Then, we denote $f_1^* \in \mathrm{L}^1(\mu)$ the unique solution. Let us define the diffusion under $\PP$:
\bea dX_t= \partial_x \ln \EE^{\PP^0}[ e^{-f^*(W^0_T)}|W^0_t=X_t] dt  + dW_t \label{SDEX} \eea  Then,
$ X_T \overset{\PP}{\sim} \mu $.
\end{prop}

\no Let us emphasize that the drift $ \partial_x \ln \EE^{\PP^0}[ e^{-f(W^0_T)}|W^0_t=X_t] $ is computed under the Wiener measure $\PP^0$.  This theorem originates from the construction of  the Schr\"{o}dinger bridge between the two marginals $\delta_{X_0}$ and $\mu$, first considered by E. Schr\"{o}dinger \cite{sch} (see the survey \cite{leo} for details and extensive references). For completeness, we report the proof in the appendix which relies on an entropy penalization, as briefly sketched below.

\subsubsection*{Entropy penalization}
We are considering a SDE of the form under $\PP$:
\beaa dX_t= \lambda_t dt  + dW_t  \eeaa   for some adapted process $\lambda_t$ properly chosen such that $X_T {\sim} \mu$. By the Girsanov theorem, the measure $\PP$ is equivalent to $\PP^0$. The calibrated measure $\PP^*$ such that $X_T \overset{\PP^*}{\sim} \mu$ can then be obtained by solving a strictly convex-constrained stochastic control problem:
\beaa P_1:=\inf_{\PP \in {\cal M}(\mu)}H(\PP|\PP^0) \eeaa where $H(\PP|\PP^0):=\EE^{\PP}[  \ln {d\PP \over d\PP^0}]$ is the relative entropy with respect to the prior $\PP^0$ (chosen here to be the Wiener measure) and ${\cal M}(\mu):=\{ \PP \sim \PP^0 \;:\; X_T \overset{\PP}{\sim} \mu\}$. The relative entropy can be replaced by arbitrary strictly convex functional - see Section \ref{fdivergence}. From the Csiszar's projection theorem (see e.g. \cite{nag}), one obtains that the infimum is attained by a unique measure $\PP^* \in {\cal M}(\mu)$.

\subsubsection*{Simplification and conditioned SDE}

\no By differentiating the (strictly concave) functional  $-\EE^{\mu}[f] -\ln \EE[ e^{-f(W^0_T)}|W^0_0=X_0]$ in (\ref{optP1}) with respect to the potential $f$, we get that the optimal potential $f^*$ is explicitly given by
\bea e^{-f^*(x)}dx =Z \mu(dx) e^{(x-X_0)^2 \over 2 T}  \label{optsol} \eea where $Z$ is an irrelevant constant, ensuring that $\mu$ has unit mass. By plugging our optimal solution (\ref{optsol}) into (\ref{SDEX}), Proposition  \ref{thm1} can then be simplified and we get
\begin{cor}[Conditioned SDE with one marginal \cite{bau}] Let us assume that $f^*$ as defined by (\ref{optsol}) is $\mu$-integrable.
Let us define under $\PP$:
\bea dX_t={ \int_\RR  {(y-X_t) \over T-t} e^{(y-X_0)^2 \over 2 T} e^{-{(y-X_t)^2 \over 2(T-t)}}   \mu(dy)
\over \int_\RR  e^{(y-X_0)^2 \over 2 T} e^{-{(y-X_t)^2 \over 2(T-t)}}   \mu(dy) }     + dW_t \label{SDEXnew} \eea  Then $ X_T \overset{\PP}{\sim} \mu $. \label{corbaudoin}
\end{cor} \no This coincides with Theorem 25 in \cite{bau} and is nothing else than a direct consequence of Schr\"{o}dinger's bridge construction \cite{sch}.

\begin{exa}[Brownian bridge] As a sanity check, by taking $\mu(x):=\delta(x-X_0)$ the Dirac mass at $X_0$, we reproduce the dynamics of a Brownian bridge:
\beaa dX_t={X_0-X_t \over T-t}dt+dW_t \eeaa
\end{exa}

\begin{rem}[With non-trivial diffusion/drift coefficients] The above construction can be trivially extended when we consider non-trivial diffusion/drift coefficients $\sigma(X_t)$ and  $b(X_t)$. We have:
\no Let us define the diffusion under $\PP$:
\beaa dX_t=   \sigma(X_t)^2 \partial_x \ln \EE^{\PP^0}[ e^{-f^*(X^0_T)}|X^0_t=X_t] dt  +b(X_t)dt+  \sigma(X_t) dW_t  \eeaa  where
under $\PP^0$: $ dX_t^0=b(X_t^0) dt +\sigma(X_t^0) dW_t^0 $. The function $f^*$ is the unique solution of a static concave optimization (assuming that $P_1 < \infty$):
\beaa P_1:=\sup_{f \in \mathrm{L}^1(\mu)} \{ -\EE^{\mu}[f] -\ln \EE^{\PP^0}[ e^{-f(X^0_T)}|X^0_0=X_0] \} \eeaa Then,
$ X_T \overset{\PP}{\sim} \mu $.
\end{rem}

\begin{rem}[Monte-Carlo Simulation] \label{Monte-Carlo Simulation} SDE (\ref{SDEX} or \ref{SDEXnew}) is highly delicate to simulate with an Euler scheme. As a numerical illustration, we have computed $C(K):=\EE^\PP[(X_t-K)^+]$ with $T=$ one year, $K \in [0.5,1.5]$ and $\mu$ a log-normal distribution with mean $X_0:=1$ and a volatility $0.3$.  The result has been quoted in terms of the Black-Scholes implied volatility and therefore should be equal to $0.3$.  One can observe that even with a timestep of $1/1000$  and $2^{18}$ MC paths, the result is still noisy (see Figure \ref{Euler}). A much better idea is to simulate under the measure $\PP^0$ under which $X_t$ is a $\PP^0$-Brownian motion. The Radon-Nikodym derivative $ {d\PP \over d\PP^0}|_{{\cal F}_t}$ is given by the Girsanov theorem:
\beaa {d\PP \over d\PP^0}|_{{\cal F}_t}&&=e^{\int_0^t \partial_x \ln \EE^{\PP^0}[ e^{-f(W^0_T)}|W^0_s=X_s] dX_s -{1 \over 2}\int_0^t (\partial_x \ln \EE^{\PP^0}[ e^{-f(W^0_T)}|W^0_s=X_s])^2 ds}\\
&&\overset{\mathrm{Ito}}{=}{\EE^{\PP^0}[ e^{-f(W^0_T)}|W^0_t=X_t] \over \EE^{\PP^0}[ e^{-f(W^0_T)}|W^0_0=X_0]  }\eeaa  where we have used that $u(t,x):=- \ln \EE^{\PP^0}[ e^{-f(W^0_T)}|W^0_t=x]$ is the solution of the Burgers PDE (see the proof of Proposition \ref{thm1} for an explanation of the appearance of this nonlinear PDE as an Bellman-Hamilton-Jacobi PDE):
\beaa \partial_t u +{1 \over 2}\partial_x^2 u -{1 \over 2}(\partial_x u)^2=0 \eeaa In particular,  the computation of a path-dependent functional $\Phi_t$ (measurable w.r.t. ${\cal F}_t$) can be written under $\PP^0$ as:
\bea \EE^\PP[ \Phi_t]&=& {  \EE^{\PP^0}[ e^{-f^*(W^0_T)} \Phi_t ]
 \over \EE^{\PP^0}[ e^{-f^*(W^0_T)}]  } \nonumber \\
 &=&{  \EE^{\PP^0}[ \mu(W^0_T) e^{(W^0_T-X_0)^2 \over 2 T} \Phi_t ]
 \over \EE^{\PP^0}[ \mu(W^0_T) e^{(W^0_T-X_0)^2 \over 2 T} ]  }   \label{EulerGirs} \eea where we have used Equation (\ref{optsol}). We have done the same experiment as above by simulating a Brownian motion and by computing (\ref{EulerGirs}) with $2^{12}$ MC paths.  As expected, we obtain a perfect match.

\begin{figure}[h]
\begin{center}
\includegraphics[width=7cm,height=5cm]{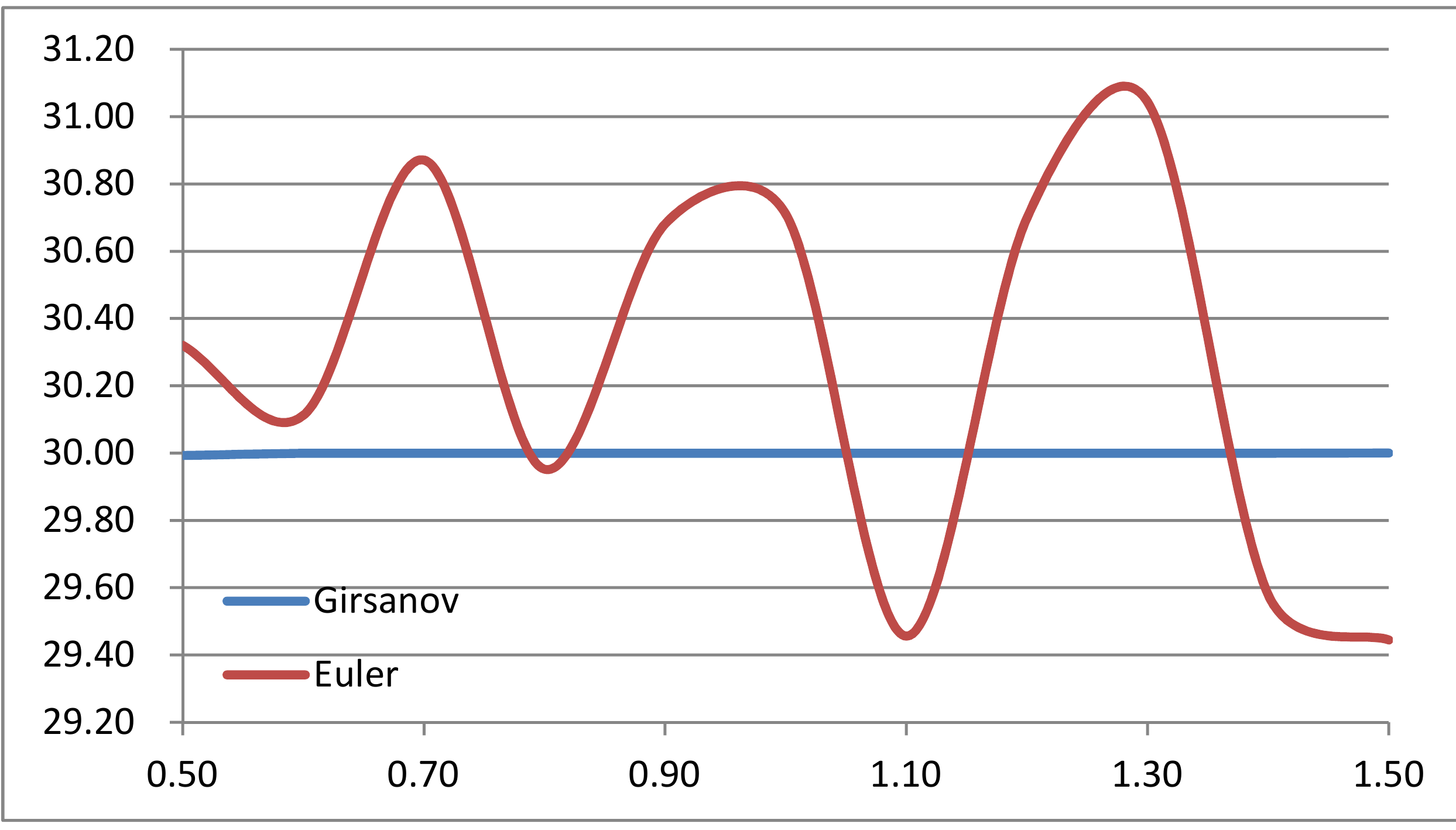}
\caption{Simulation of SDE (\ref{SDEX}  or \ref{SDEXnew}) using (1) an Euler scheme with a timestep of $1/1000$  and $2^{18}$ MC paths and (2) using the Girsanov transform (\ref{EulerGirs}) with $2^{12}$ MC paths. We have quoted $C(K):=\EE^\PP[(X_T-K)^+]$ in terms of the Black-Scholes implied volatility ($\times 100$) for different strikes $K$ with $T=$ one year and $\mu$ a log-normal distribution with mean $X_0:=1$ and volatility $0.3$. $K$ ranges in $[0.5,\cdots,1.5]$.}
\label{Euler}
\end{center}
\end{figure}
 \label{Monte-Carlo Simulation}
\end{rem}

\subsection{Density factorization and Doob's transform}
The density $p(t,x|X_0)$ of $X_t$ defined by SDE (\ref{SDEX}) can be factorized (highlighting a striking relation with Born's rule in quantum mechanics -- see \cite{nag} for an extensive discussion) as
\bea p(t,x|X_0)=\Psi(t,x) \bar{\Psi}(t,x) \label{fact} \eea where $\Psi(t,x)$ (resp. $\bar{\Psi}(t,x)$) is the solution of the backward (resp. forward) PDE:
\beaa \partial_t  \Psi(t,x)+{\cal L}\Psi(t,x)&=&0, \quad \Psi(T,x)=e^{-f^*(x)} \\
\partial_t  \bar{\Psi}(t,x)&=&{\cal L}^\dag\bar{\Psi}(t,x),  \quad \lim_{t \rightarrow 0} \Psi(t,x) \bar{\Psi}(t,x)=\delta(x-X_0) \eeaa with ${\cal L}:={1  \over 2}\partial_x^2$. \no  Indeed, one can check that $p(t,x|X_0)$ satisfies the Fokker-Planck PDE as required with
\beaa  \Psi(t,x)=\EE^{\PP^0}[ e^{-f^*(W^0_T)}|W^0_t=x]  \eeaa  This implies that $\bar{\Psi}(0,x)={ \delta(x-X_0)  \over \EE^{\PP^0}[ e^{-f^*(W^0_T)}|W^0_0=X_0]}$ and
\beaa \bar{\Psi}(t,x)={ p_0(t,x|X_0) \over \EE^{\PP^0}[ e^{-f^*(W^0_T)}|W^0_0=X_0]} \eeaa with $p_0(t,x|X_0)=e^{-(x-X_0)^2/2t}/\sqrt{2 \pi t}$. Finally, $p(t,x|X_0)$ can be written as
\beaa p(t,x|X_0)=\EE^{\PP^0}[ e^{-f^*(W^0_T)}|W^0_t=x] p_0(t,x|X_0) {1 \over \EE^{\PP^0}[ e^{-f^*(W^0_T)}|W^0_0=X_0]} \eeaa
 In probability terms, this factorization corresponds to a Doob's $\Psi$-transform applied to the prior $\PP^0$.

\subsection{Multi-marginals}
By using the Markov property of $X_t$, Proposition \ref{thm1} can be easily generalized in the case of multi-marginals. We consider again the optimization problem: \beaa P_n:=\inf_{\PP \in {\cal M}((\mu_i)_{1 \leq i \leq n})}H(\PP|\PP^0) \eeaa where
${\cal M}((\mu_i)_{1 \leq i \leq n}):=\{ \PP \sim \PP^0 \;:\; X_{t_i} \overset{\PP}{\sim} \mu_i, \quad i=1,\cdots,n\}$ with $t_1<t_2<\cdots <t_n$. For example, in the case of two marginals $\mu_1$ and $\mu_2$ (the extension to $n$ marginals is straightforward - see Corollary \ref{Conditioned SDE with two marginals: again}), we obtain
\begin{prop}[Two marginals] \label{thm1bis} Let us consider the static strictly concave optimization:
\bea P_2:=\sup_{f_1 \in \mathrm{L}^1(\mu_1), f_2 \in \mathrm{L}^1(\mu_2)} \Phi(f_1,f_2), \quad \label{P}
\\ \Phi(f_1,f_2):= -\EE^{\mu_1}[f_1]-\EE^{\mu_2}[f_2] -\ln \EE^{\PP^0}[ e^{-f_1(W^0_{t_1})-f_2(W^0_{t_2})}|W^0_0=X_0] \nonumber \eea and assume that $P_2<\infty$.
We denote  $f^*_1,f^*_2$  the unique solutions. Let us define under $\PP$:
\beaa dX_t= \partial_x \ln \EE^{\PP^0}[ e^{-f^*_1(W^0_{t_1})1_{t\leq t_1}-f^*_2(W^0_{t_2})}|W^0_t=X_t] dt  + dW_t \eeaa  Then,
$ X_{t_1} \overset{\PP}{\sim} \mu_1,\quad  X_{t_2} \overset{\PP}{\sim} \mu_2 $.
\end{prop} \no $f_1$ and $f_2$ are called the Schr\"{o}dinger potentials and can be related to the Monge-Kantorovich potentials by considering the entropic relaxation of an optimal transportation problem.  A similar factorization (\ref{fact}) holds with
\beaa  \Psi(t,x)=\EE^{\PP^0}[ e^{-f^*_1(W^0_{t_1})1_{t\leq t_1}-f^*_2(W^0_{t_2})}|W^0_t=X_t] \eeaa

\begin{rem}[Sinkhorn's algorithm] \no By differentiating the functional  $\Phi(f_1,f_2)$ with respect to the potentials $f_1$ and $f_2$, we get that the optimal potentials $f_1^*,f_2^*$ are given by
\bea Z \mu_1(x)&=&e^{-f^*_1(x)} e^{-{(x-X_0)^2 \over 2t_1}}\int_\RR e^{-{(y-x)^2 \over 2(t_2-t_1)}} e^{-f^*_2(y)} dy \label{sink1}\\
Z \mu_2(x)&=&e^{-f^*_2(x)}\int_\RR e^{-{(y-x)^2 \over 2(t_2-t_1)}} e^{-{(y-X_0)^2 \over 2t_1}} e^{-f^*_1(y)} dy \label{sink2}
\eea with $Z$ an irrelevant constant.  The static optimization problem (\ref{P}) can be solved using the Sinkhorn algorithm which consists in doing sequentially the two  iterations (\ref{sink1},\ref{sink2}), leading to a  convergence with a linear convergence  rate. Note in particular that $f_1^*(x)$ (resp. $f_2^*(x)$) is explicitly fixed if $f_2^*(x)$ (resp. $f_1^*(x)$) is given (see Equations (\ref{sink1},\ref{sink2})).
\end{rem}
\no Using this explicit expression of $e^{-f^*_1(x)}$ as a function of $e^{-f^*_2(\cdot)}$ (Equation (\ref{sink1})), Proposition  \ref{thm1bis} can then be simplified and we get

\begin{cor}[Conditioned SDE with two marginals: again]\label{Conditioned SDE with two marginals: again} Let us consider the static strictly concave optimization
\beaa P_2:=\sup_{ f_2 \in \mathrm{L}^1(\mu_2)}
-\EE^{\mu_2}[f_2] -\int \mu_1(dx) \ln \EE^{\PP^0}[ e^{-f_2(W^0_{t_2})}|W^0_{t_1}=x] \nonumber \eeaa and assume that $P_2<\infty$.  We denote $f_2^*$ the unique solution. Let us define under $\PP$ for all $t\in(t_1,t_2]$:
\beaa dX_t= \partial_x \ln \EE^{\PP^0}[ e^{-f^*_2(W^0_{t_2})}|W^0_t=X_t] dt  + dW_t \eeaa \no Then,
$  X_{t_2} \overset{\PP}{\sim} \mu_2 $. Similarly, $f_2^*$ is the unique solution of the nonlinear equation:
\beaa \mu_2(y)= e^{-f^*_2(y)} \int_\RR { e^{-{(y-x)^2 \over 2(t_2-t_1)}} \mu_1(dx) \over
\int_\RR e^{-{(z-x)^2 \over 2(t_2-t_1)}} e^{-f^*_2(z)} dz } \eeaa
\end{cor} \no By iteratively applying this construction over the intervals $[t_{i-1},t_{i}]$, we obtain a bridge $\PP \in {\cal M}((\mu_i)_{1 \leq i \leq n})$.

\begin{rem}[Limit $\Delta t:=t_2-t_1 \rightarrow 0$] We assume that $\mu_2(x)=\mu_1(x)+\Delta t \partial_x\{ \mu_1(x) \psi_{12}(x)\}+o(\Delta t)$. Then in the first-order in $\Delta t$, we have
\beaa f_1(x) &=& -  \psi_{12}(x)- {(x-X_0)\over t_1} -{1 \over 2}\partial_x \ln \mu_1(x) \\
f_2(x)&=&A(x)+  \frac{1}{2} \left(A'(x)^2-A''(x)\right)\Delta t \\
A(x)&=&-f_1(x)-\ln \left(Z \mu_1(x)\right)-{(x-{X_0})^2 \over 2 t_1} \eeaa
\end{rem}

\subsection{Decoupling and relative entropy} \label{Decoupling}
A close inspection of Propositions  (\ref{thm1}) and (\ref{thm1bis}) reveals that the computing of the drift for the calibrated diffusion is obtained by solving first a static concave optimization -- similar to the entropic  construction of Vanilla smiles in  \cite{demarch} and then the computation of a conditional expectation (under the measure $\PP^0$). This decoupling can be directly justified by observing that the  relative entropy $H(\PP|\PP^0)$ can be disintegrating (i.e., by taking the conditional expectations) with respect to $X_{t_1}:=x_1$ and $X_{t_2}:=x_2$ and we have:
\bea H(\PP|\PP_0)=\int_{\RR^2} H(\PP^{x_1,x_2}|\PP_0^{x_1,x_2}) p(x_1,x_2) dx_1 dx_2+\int_{\RR^2} p(x_1,x_2)\ln {p(x_1,x_2) \over p_0(x_1,x_2)} dx_1dx_2  \label{decouplingentropy} \eea We deduce that the optimal value of $P_2$ is attained by
\beaa  \PP^{x_1,x_2}=\PP_0^{x_1,x_2} \eeaa and $p(x_1,x_2)$ is the (dual) solution of the above static  concave optimization (\ref{P}). Note that we could have consider $f$-divergence  instead of $H(\PP|\PP^0)$, see Section \ref{fdivergence}. However, the decoupling property is no more satisfied, highlighting the choice of the relative entropy as a convenient strictly convex function.

\subsection{Infinitely-many marginals}

Let us define $t$-marginals $\mu_t$ for all $t>0$ and set $F_\mu(t,K):=\EE^{\mu_t}[1_{X_t \leq K}]$.

\begin{prop}[Infinitely-many marginals] \label{thm1bisbis}
Let us define under $\PP$:
\beaa dX_t={ {1 \over 2}\partial_x^2 F_\mu(t,X_t)-\partial_t  F_\mu(t,X_t) \over \partial_x F_\mu(t,X_t)}dt + dW_t \eeaa and assume that the SDE is well-posed. Then $X_t \overset{\PP}{\sim} \mu_t$ for all $t > 0$.
\end{prop}

\subsection{$f$-divergence} \label{fdivergence}
\no The Schr\"{o}dinger construction can be generalized by replacing the entropy distance by the $f$-divergence:
\beaa K_f(\PP|\PP^0):=\EE^{\PP^0}[ f\left({d\PP \over d\PP^0}\right)] \eeaa where  $f$ is a strictly convex function with $f(1):=0$. The relative entropy corresponds to take $f(m):=m \ln m$. We obtain:
\begin{prop}[$f$-divergence, one marginal] Let us consider the static strictly concave optimization:
\bea P_1:=\sup_{f_1 \in \mathrm{L}^1(\mu)} \{ -\EE^{\mu}[f_1] +u(0,X_0,1) \}\eea and assuming that $P_1<\infty$.  We denote  $f_1^*$  the unique solution. Let us define the diffusion under $\PP$:
\beaa dX_t&=&\lambda_t^* dt  + dW_t  \\
dM_t&=&\lambda_t^* M_t dW_t, \quad M_{t=0}=1  \eeaa  where
\beaa \lambda_t^*:=-{\partial_{mx} u(t,X_t,M_t) \over M_t \partial_{m}^2 u(t,X_t,M_t)} \eeaa and
$u(t,x,m)$ is the unique solution of
\bea \partial_t u+{1\over 2}\partial_x^2 u -{1 \over 2} {(\partial_{xm} u)^2 \over \partial_{mm} u}=0, \quad u(T,x,m)=f(m)+f_1^*(x)m \label{PDEdiv} \eea  Then,
$ X_T \overset{\PP}{\sim} \mu $.
\end{prop}

\begin{exa}[$\chi^2$-divergence] Let us consider the $\chi^2$-divergence for which $f(m)=m^2-1$. The solution of PDE (\ref{PDEdiv}) is then
\beaa u(t,x,m)&=&m^2+A(t,x)m+C(t,x) \\
A(t,x)&=&\EE^{\PP^0}[f_1^*(W_T^0)|W_t^0=x] \\
\partial_t C +{1 \over 2}\partial_{xx} C-{(\partial_x A)^2 \over 4}&=&0, \quad C(T,x)=-1\eeaa with $\PP^0$ the Wiener measure. The drift is
$ \lambda_t^*:=-{\partial_x A(t,X_t) \over 2 M_t } $.
\end{exa}

\subsection{Brownian with a prescribed marginal and a fixed area}
The above construction can be extended if we impose to the process $X_t$ to have a fixed area $\EE[ \int_0^T X_s ds]:={\cal A}$ and a prescribed $T$-marginal $\mu$. In mathematical finance, this can be interpreted as giving the prices of Vanillas and a forward on an Asian option.

\no Let us define the diffusion under $\PP$:
\beaa dX_t&=& \partial_x \ln \EE^{\PP^0}[ e^{-f^*(W^0_T)-A^* \int_t^T W^0_s ds}|W^0_t=X_t] dt  + dW_t \eeaa  where
the function $f^*$ and the number $A^* \in \RR$ are the unique solutions of a static concave optimization:
\bea P_1:=\sup_{f \in \mathrm{L}^1(\mu), A \in \RR} \{ -\EE^{\mu}[f]-A \; {\cal A} -\ln \EE^{\PP^0}[ e^{-f(W^0_T)-A \int_0^T W^0_s ds}|W^0_0=X_0] \}\label{optP1} \eea Then, $X_T \overset{\PP}{\sim} \mu$ and $\EE^\PP[\int_0^T X_s ds]={\cal A}$. The proof is similar to the proof of Proposition \ref{thm1} and is therefore not reproduced. The infimum over $f$ and $A$ can be computed and we get
\beaa \mu(x)={e^{-f^*(x)} \PP^0(W_T=x) \over Z}\EE^{\PP^0}[ e^{-A^* \int_0^T W_s ds}   |W_T=x] \eeaa and
\beaa {\cal A}=\int \mu(dx) {\EE^{\PP^0}[ e^{-A^* \int_0^T W_s ds} \int_0^T W_s ds |W_T=x] \over \EE^{\PP^0}[ e^{-A^* \int_0^T W_s ds}  |W_T=x] } \eeaa
Using that the   $(W_T,\int_0^T W_s ds)$  is a two-dimensional Gaussian vector with covariance $\left(
\begin{array}{cc}
 T & \frac{T^2}{2} \\
 \frac{T^2}{2} & \frac{T^3}{3} \\
\end{array}
\right)$ and mean $(X_0,X_0T)$, we get:
\bea A^*&=&{6 \over T^2}\left( X_0+\int y \mu(dy) -{2 {\cal A} \over T} \right) \label{defnA} \\
 \mu(x)&=&{e^{-f^*(x)}   \over Z} \frac{e^{\frac{(A^*)^2 T^3}{24}-\frac{1}{2} A^* T
   (x+X_0)-\frac{(x-X_0)^2}{2 T}}}{\sqrt{2 \pi  T}
   } \label{eqmu}  \eea Finally, this implies
   \begin{prop} Let us assume that $f^*$ as defined (\ref{eqmu})  is $\mu$-integrable and $\int |y| \mu(dy)<\infty$. Let us define under $\PP \sim \PP^0$:
    \beaa dX_t&=& \partial_{X_t} \ln
   \int \mu(dx)  e^{\frac{1}{24} \left((A^*)^2 (T-t)^3-(A^*)^2 T^3-12 A^* (T-t)
   (x+X_t)+12 A^* T (x+X_0)-\frac{12
   (x-X_t)^2}{T-t}+\frac{12 (x-X_0)^2}{T}\right)}  dt \\
   &+& dW_t \eeaa with $A^*$ defined by (\ref{defnA}). Then,
   $X_T \overset{\PP}{\sim} \mu$ and $\EE^\PP[\int_0^T X_s ds]={\cal A}$.
   \end{prop}

\section{A new class of SVMs matching Vanillas: Martingale Schr\"{o}dinger bridges} \label{A new class of SVMs matching Vanillas}

\subsection{Naked SVM}
\no Let us consider a naked SVM defined under a risk-neutral measure $\PP^0$ by
\beaa dS_t&=&S_t a_t dW^0_t,  \quad d\langle Z^0, W^0\rangle_t=\rho dt,\quad \rho \in (-1,1) \\
da_t&=&b(a_t)dt+ \sigma(a_t) dZ^0_t \eeaa Let us emphasize again that although for the sake of simplicity, we consider one-dimensional factor SVMs, our results extend  to  multi-dimensional SVMs (see however our discussion on the numerical implementation which is more involved from a multi-dimensional SVM). Furthermore, we could assume that $b$ and $\sigma$ depend also on $S$ although common SVMs, used by practitioners, do not assume such a dependence.

\no We denote below ${\cal L}^0$ the It\^o generator of the process $(S_t,a_t)$ and consider one-dimensional marginals $(\mu_i)_{1 \leq i \leq n}$ increasing in the convex order, meaning that for all convex functions $f$:
\beaa \EE^{\mu_i}[f] \leq  \EE^{\mu_{i+1}}[f],\quad \forall i=1,\cdots,n-1 \eeaa In practice, the marginals $(\mu_i)_{1 \leq i \leq n}$  are implied from  market values (at $t=0$) of $t_i$-Vanilla options $C(t_i,K)$: $\mu_i(K):=\partial_K^2C(t_i,K)$.  $\cal A$ denotes the space of adapted process in $\mathrm{L}^2(\PP^0)$.

\subsection{One marginal}

\begin{thm}[One marginal] \label{thm2} Let us consider the strictly concave  optimization problem:
\bea P_1:=\sup_{f_1\in \mathrm{L}^1(\mu_1), \Delta_\cdot \in {\cal A}}\{ - \EE^{\mu_1}[f_1] -\ln \EE^{\PP^0}[e^{-f_1(S_{t_1})-\int_0^{t_1} \Delta_s dS_s}]  \}
 \label{opt1} \eea and assume that $P_1< \infty$. We denote   $f^*_1$ and  $(\Delta^*_s)_{s \in (t_0,t_1)}$ the unique solutions. Let us consider the SVM defined under a measure $\PP$ by
\bea dS_t&=&S_t a_t dW_t,  \quad d\langle Z, W\rangle_t=\rho dt \label{SVMnew0} \\
da_t&=&\left(b(a_t)+(1-\rho^2)\sigma(a_t)^2 \partial_a  \ln \EE^{\PP^0}[ e^{-f^*_1(S_{t_1})-\int_t^{t_1} \Delta_s^* dS_s}|S_t,a_t] \right)dt+ \sigma(a_t) dZ_t  \nonumber  \eea for all $t \in [t_0:=0, t_1]$.

 \no {\bf (1)}  Then,
$ S_{t_1} \overset{\PP}{\sim} \mu_1 $.

\no {\bf (2)} $\Delta_t^*$ is given by \bea
\Delta_t^*=-\left( \partial_s +\rho {\sigma(a_t) \over S_t a_t} \partial_a \right)u(t,S_t,a_t) \label{Delta} \eea where $u$ is the unique solution of the Burgers-like semi-linear PDE:
\bea \partial_t u + {\cal L}^0u  - {1 \over 2 }(1-\rho^2) (\sigma(a) \partial_a u)^2=0, \quad u(t_1,s,a):=f^*_1(s) \label{burgersSVM}\eea and the optimal potential is
\beaa e^{-f^*_1(K)}=Z {\mu_1(K) \over  \PP^0(S_{t_1}=K) \EE^{\PP^0}[ e^{-\int_0^{t_1} \Delta_s^* dS_s}|S_{t_1}=K] } \eeaa

\no  {\bf (3) } An equivalent formulation of the optimization $P_1$ is
\beaa P_1=\sup_{f_1\in \mathrm{L}^1(\mu_1)}\{ - \EE^{\mu_1}[f_1] +u(0,S_0,a_0) \}
 \eeaa

\end{thm}

\noindent Our proofs are reported in the appendix.

\begin{rem}[Finite number of strikes] In practice, only a finite number of calls with strikes $K_1 < K_2 < \cdots < K_N$ are quoted in the market. Instead of calibrating the full marginal $\mu_1$ (which is unknown), we want to match
\beaa C(t_1,K_\alpha):=\EE^\PP[(S_{t_1}-K_\alpha)^+],\quad \forall \alpha=1,\cdots,N \eeaa Our theorem (\ref{thm2}) still applies where the potential $f_1$ is now restricted to be of the form:
\beaa f_1(s):=\sum_{\alpha=1}^N  \omega_\alpha (s-K_\alpha)^+\eeaa for some real parameters $(\omega_\alpha)_{1 \leq \alpha \leq N}$. \label{Finite number of strikes}
\end{rem}

\begin{rem}[Monte-Carlo Simulation] Instead of simulating under $\PP$, it is better to simulate under $\PP^0$, see previous Remark \ref{Monte-Carlo Simulation}.  The Radon-Nikodym derivative ${d\PP \over d\PP^0}|_{{\cal F}_t}$ is
\beaa {d\PP \over d\PP^0}|_{{\cal F}_t}={\EE^{\PP^0}[ e^{-f^*_1(S_{t_1})-\int_0^{t_1} \Delta^*_s dS_s}|{\cal F}_t]  \over
\EE^{\PP^0}[ e^{-f^*_1(S_{t_1})-\int_0^{t_1} \Delta^*_s dS_s}]}  \eeaa  Therefore, for a $\FF_t$-measurable payoff with $t \leq t_1$, we have
\beaa \EE^\PP[ \Phi_t ]={\EE^{\PP^0}[ \Phi_t  e^{-f^*_1(S_{t_1})-\int_0^{t_1} \Delta^*_s dS_s}] \over
\EE^{\PP^0}[ e^{-f^*_1(S_{t_1})-\int_0^{t_1} \Delta^*_s dS_s}]}  \eeaa
\end{rem}

\subsection{Numerical implementations} \label{Numerical implementations}

\no Our numerical algorithm can be described by the following steps:
\begin{enumerate}
\item Using a Newton gradient descent algorithm, we solve the concave problem:
\beaa \sup_{f_1\in \mathrm{L}^1(\mu_1)}\{ - \EE^{\mu_1}[f_1] +u(0,S_0,a_0) \} \eeaa At each step of the gradient iteration,  $u$ is obtained by solving the PDE  (\ref{burgersSVM}). We finally store $\Delta^*_s$ for $s\in [0,t_1)$ as given by Equation (\ref{Delta}) when the algorithm has converged. In practice, following Remark \ref{Finite number of strikes}, $f$ is decomposed over calls and the optimisation over $f_1$ is replaced by an optimization over $\omega \in \RR^N$. The gradient with respect to $\omega_\alpha$ is then given by
\beaa  - \EE^{\mu_1}[(S_{t_1}-K_\alpha)^+] +\Delta_\alpha(0,S_0,a_0)  \eeaa with $\Delta_\alpha$ the solution of the linear PDE:
\beaa \partial_t \Delta_\alpha + {\cal L}^0\Delta_\alpha  - (1-\rho^2) \sigma(a)^2 \partial_a u \partial_a \Delta_\alpha=0, \quad \Delta_\alpha(t_1,s,a)=(s-K_\alpha)^+ \eeaa
\item Finally, the pricing of an option with payoff $\Phi_{t_1}$ is performed by MC under the measure $\PP^0$ using
\beaa \EE^{\PP}[\Phi_{t_1}]={\EE^{\PP^0}[e^{-f^*_1(S_{t_1})-\int_0^{t_1} \Delta^*_s dS_s} \Phi_{t_1}] \over \EE^{\PP^0}[e^{-f_1(S_{t_1})-\int_0^{t_1} \Delta^*_s dS_s} ] }\eeaa
\end{enumerate}

\begin{exa}[Numerical examples on \texttt{TOTAL}]
We have checked our algorithm for  \texttt{TOTAL} Vanillas, pricing-date = 9/12/2018 and maturity =1.2 years. The naked SVM has been chosen to be a (log-normal) SABR model with $\alpha=15.8\%$, $\nu=40\%$  and $\rho=-61\%$ (see the smile denoted ``Naked SABR'' in Figure \ref{Vanillacheck}). In particular, the smile as produced by our naked SABR matches the at-the-money volatility but has an incorrect skew, due to our choice of the spot-volatility correlation.  Once the drift has been calibrated using the algorithm outlined above, we have repriced the Vanillas by Monte-Carlo and compare with the market prices.  We reproduce the market smile (see Figure \ref{Vanillacheck} where the blue and green curves coincide).  We give also the optimized weights $(\omega^*_\alpha)_{1 \leq \alpha \leq 20}$.

\begin{figure}[h]
\begin{center}
\includegraphics[width=7cm,height=5cm]{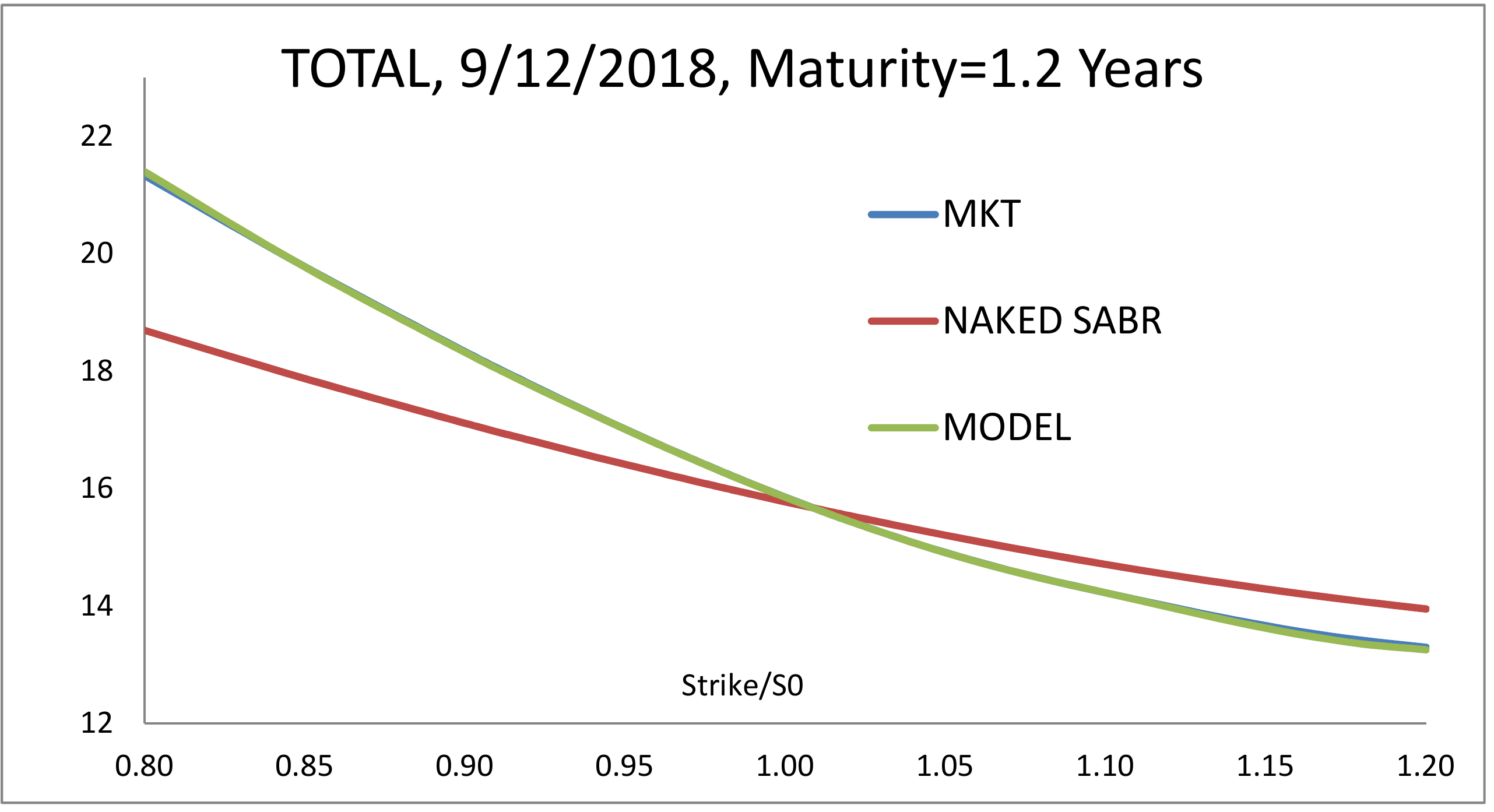}
\includegraphics[width=7cm,height=5cm]{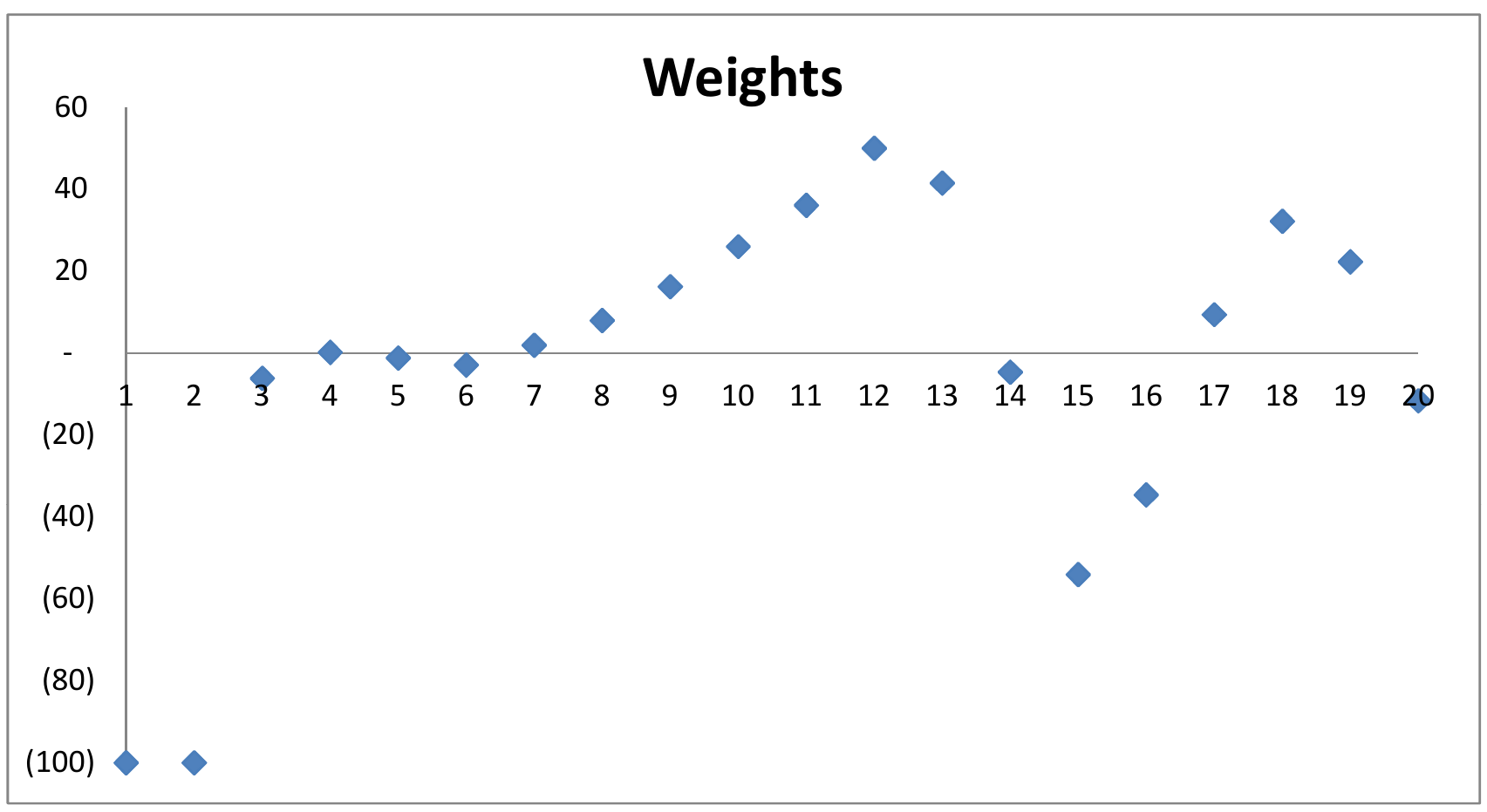}
\end{center}
\caption{Calibration to Vanillas \texttt{TOTAL}, pricing-date = 9/12/2018. Maturity=1.2 years. Left: Smiles. Right: optimized weights $(\omega^*_\alpha)_{1 \leq \alpha \leq 20}$ for the twenty strikes.}
 \label{Vanillacheck}
\end{figure}

\end{exa}

\subsubsection{Alternative}

\no The above algorithm requires to solve the non-linear Burgers-like PDE  (\ref{burgersSVM}). Below, we present an alternative algorithm, which requires only to solve linear PDEs.
\begin{enumerate}
\item By MC, simulate and store $N_\mathrm{MC}$ MC paths of $S_{t_1}$ under $\PP^0$. Compute and store $\PP^0(S_{t_1}=K)$  for different values of $K$ belonging to a one dimensional grid.
\item  Set $n:=1$. Set also $\Delta_s^{(0)}:=0$ and
\beaa e^{-f^{(0)}_1(K)}:= Z{\mu_1(K) \over \PP^0(S_{t_1}=K)} \eeaa
\item Compute $ \Delta_t^{(n)}:=\left( \partial_s +\rho {\sigma(a_t) \over S_t a_t} \partial_a \right) \ln U^n(t,s,a)$ with \beaa U^n(t,s,a):=\EE^{\PP^0}[ e^{-f^{(n-1)}_1(S_{t_1})-\int_t^{t_1} \Delta_s^{(n-1)} dS_s}|S_t,a_t] \eeaa \no by solving a {\it linear} parabolic PDE:
\beaa \partial_t U^{(n)} + {\cal L}^0 U^{(n)} + {1 \over 2} s^2 a^2 (\Delta^{n-1})^2 U^{(n)} -{1 \over 2}  s^2 a^2
(\Delta^{n-1}) \left( \partial_s +\rho {\sigma(a) \over s a } \partial_a \right)  U^{(n)}=0,  \\ \quad    U^{(n)}(t_1,s,a)=e^{-f^{(n-1)}_1(s)} \eeaa
 and set
\beaa e^{-f^{(n)}_1(K)}:= Z {\mu_1(K) \over \PP^0(S_{t_1}=K)}{1 \over \EE^{\PP^0}[  e^{-\int_0^{t_1} {\Delta}^{(n)}_t dS_t} |S_{t_1}=K] } \eeaa using a MC simulation.
\item Set $n:=n+1$ and iterate Step (3-4) up to convergence.
\end{enumerate}

\subsection{Multi-marginals}

\begin{thm}[Marginals $\mu_1$ and $\mu_2$]  \label{thm2bis} Let us consider the strictly concave  optimization problem:
\bea P_{12}:=\sup_{f_2 \in \mathrm{L}^1(\mu_2) ,\Delta_{s \in (t_1,t_2)} \in {\cal A}} -\EE^{\mu_2}[f_2] -\EE^{\PP}[\ln \{ \EE^{\PP^0}[e^{- f_2(S_{t_2})-\int_{t_1}^{t_2}\Delta_s dS_s}|S_{t_1},a_{t_1}] \} ]  \label{opt2} \eea and assume that $P_{12}<\infty$. We denote $f^*_2$, $(\Delta_s^*)_{s \in (t_1,t_2)}$ the unique solutions.  Let us consider the SVM defined under a measure $\PP$ by
\bea dS_t&=&S_t a_t dW_t,  \quad d\langle Z, W\rangle_t=\rho dt  \label{SVMnew} \\
da_t&=&\left(b(a_t)+(1-\rho^2)\sigma(a_t)^2 \partial_a \ln \EE^{\PP^0}[ e^{-f^*_2(S_{t_2})-\int_t^{t_2}\Delta^*_s dS_s} |S_t,a_t]\right)dt+ \sigma(a_t) dZ_t \nonumber  \eea
 for all $t \in [t_1, t_2]$ and SDE (\ref{SVMnew0}) for $[t_0,t_1]$.

\no {\bf (1)} Then, $ S_{t_i} \overset{\PP}{\sim} \mu_i, \quad i=1,2$.

\no  {\bf (2) } An equivalent formulation of the optimization $P_{12}$ is
\beaa P_{12}=\sup_{f_2\in \mathrm{L}^1(\mu_2)}\{ - \EE^{\mu_2}[f_2] +\EE^\PP[ u(t_1,S_{t_1},a_{t_1})] \}
 \eeaa where $u$ is the unique solution of the Burgers-like PDE:
\bea \partial_t u + {\cal L}^0u  - {1 \over 2 }(1-\rho^2) (\sigma(a) \partial_a u)^2=0, \quad u(t_2,s,a)=f_2(s)\eea

\no  {\bf (3) } An equivalent formulation of the optimization $P_{12}$  is also
\beaa P_{12}:=\sup_{f_i \in \mathrm{L}^1(\mu_i),\Delta_{s \in (0,t_2)} \in {\cal A}}\{ - \sum_{i=1}^2 \EE^{\mu_i}[f_i]-\ln \EE^{\PP^0}[e^{-\sum_{i=1}^2 f_i(S_{t_i})-\int_{0}^{t_2}\Delta_s dS_s}] \}   \eeaa

\end{thm}

\no By sequentially applying Theorem \ref{thm2bis}, as our SVM is a Markov process, one can then construct a SVM, obtained by concatenation of diffusions (\ref{SVMnew}) (and (\ref{SVMnew0}) for the interval $[t_0,t_1]$) over the intervals $[t_i,t_{i+1}]$ ($i=0,\cdots,n-1$),  such that $ S_{t_i} \overset{\PP}{\sim} \mu_i, \quad i=1,\cdots,n $. By construction, this SVM leads also to a convex-order interpolation of the marginals $(\mu_i)_{1 \leq i \leq n}$ as $S_t$ is a martingale:
\begin{cor}[Entropic convex-order interpolation] Under $\PP$,
\beaa \mu_{i-1} \overset{\mathrm{conv}}{\leq} \mathrm{Law}(S_t) \overset{\mathrm{conv}}{\leq} \mu_{i},\quad \forall t  \in [t_{i-1},t_i],\quad \forall i=1,\cdots,n   \eeaa
\end{cor}

\subsection{Numerical implementation}

\no The numerical algorithm can be described by the following steps which combine a Monte-Carlo simulation and a PDE solver:
\begin{enumerate}
\item Once the model between $[0,t_1]$ has been calibrated following the numerical method described in Section \ref{Numerical implementations}, we simulate and store $N_\mathrm{MC}$ Monte-Carlo paths $(S^{(i)}_{t_1},a^{(i)}_ {t_1})_{1 \leq i \leq N_\mathrm{MC}}$ under $\PP^0$ and also store  the Radon-Nikodym derivative $(G_i)_{1 \leq i \leq N_\mathrm{MC}}$ for each MC paths:
\beaa G_{t_1}^{(i)}:={e^{-f^*_1(S^{(i)}_{t_1})-\int_0^{t_1} \Delta_s^{*,(i)} dS^{(i)}_s} \over e^{-u(t_0,S^{(i)}_{t_0},a^{(i)}_{t_0})}}  \eeaa In practice, the It\^o integral $\int_0^{t_1} \Delta_s^* dS_s$ is discretized using an Euler scheme.
\item Using a Newton gradient descent algorithm, we solve the concave problem (over $\RR^N$):
\beaa \sup_{f_2\in \mathrm{L}^1(\mu_2)}\{ - \EE^{\mu_1}[f_2] + {1 \over N_\mathrm{MC}} {\sum_{i=1}^{N_\mathrm{MC}} G_{t_1}^{(i)} u(t_1,S^{(i)}_{t_1},a^{(i)}_{t_1})
 } \} \eeaa At each step of the gradient iteration,  $u$ is obtained by solving the PDE  (\ref{burgersSVM}) between $[t_1,t_2]$. In practise the distance between two Vanilla maturities is less than 6 months and therefore the numerical solution of the Burgers PDE is fast. This is identical to Step (1) in Section \ref{Numerical implementations}. We store $\Delta_s^*$ for $s \in [t_1,t_2)$.
\item Simulate and store $N_\mathrm{MC}$ Monte-Carlo paths $(S^{(i)}_{t_2},a^{(i)}_ {t_2})_{1 \leq i \leq N_\mathrm{MC}}$ under $\PP^0$ and update  the Radon-Nikodym derivative $(G^{(i)})_{1 \leq i \leq N_\mathrm{MC}}$  for each MC paths at $t_2$ by
\beaa G_{t_2}^{(i)}:= G_{t_1}^{(i)} \times {e^{-f^*_2(S^{(i)}_{t_2})-\int_{t_1}^{t_2} \Delta_s^{*,(i)} dS^{(i)}_s} \over e^{-u(t_1,S^{(i)}_{t_1},a^{(i)}_{t_1})}} \eeaa
\item Iterate Steps (2-3) until the last maturity $t_n$.
\item Finally the (undiscounted) price of an option with path-dependent payoff $\Phi(S_{t_1},S_{t_2},\cdots,S_{t_n})$ is given for $N_\mathrm{MC}$ large  by
\beaa  \EE^{\PP}[\Phi(S_{t_1},S_{t_2},\cdots,S_{t_n})] \approx {1 \over N_\mathrm{MC}} \sum_{i=1}^{ N_\mathrm{MC}} G_{t_n}^{(i)} \Phi(S^{(i)} _{t_1},S^{(i)} _{t_2},\cdots,S^{(i)} _{t_n})\eeaa
\end{enumerate}

\begin{rem}[Multi-dimensional SVM] Our algorithm requires to solve a two-dimensional nonlinear Burgers PDE. In the case of a multi-factor SVM, the numerical solution could not be obtained using a deterministic scheme which suffers from the curse of dimensionality. One possible tentative could be to use a Monte-Carlo algorithm for solving the Burgers PDE as described in \cite{tal}. We left this for future research.
\end{rem}

\section{A new class of SVMs matching Market instruments} \label{A new class of SVMs matching Market instruments}

\no In this section, we generalize our previous results where  we replace the Vanillas by some other market instruments. Instead of presenting the extension  in full generality with arbitrary market instruments, we give three (relevant) examples. The first considers options written on 	realized variance, the second one options on running maximum and the last one both Vanillas and options on VIX.

\subsection{SVM matching options on variance}

\no Let us consider a naked SVM for which the instantaneous volatility is under  $\PP^0$:
\beaa da_t= b(a_t)dt+\sigma(a_t) dZ_t^0 \eeaa For use below, we denote $V_t:=\langle \ln S \rangle_t=\int_0^t a_s^2 ds$. As in previous sections, we want to add a drift term such that we calibrate (at $t=0$)  market prices of options on variance $C_\mu^\mathrm{mkt}(K)$ written on the quadratic variation  $V_T$, meaning that
\beaa \EE^\PP[\left(V_T -K\right)^+]=C_\mu^\mathrm{mkt}(K),\quad \forall K \in \RR_+ \eeaa By differentiating twice with respect to $K$, this is equivalent to impose the marginal of $V_T$:
\beaa V_T\overset{\PP}{\sim}\mu:=\partial_K^2 C_\mu^\mathrm{mkt} \eeaa

\begin{thm}\label{thm3} Let us consider the strictly concave optimization
 problem:
\beaa P_1:=\sup_{f \in \mathrm{L}^1(\mu)} -\EE^\mu[f]-\ln \EE^{\PP^0}[  e^{-f(V_T)} ] \eeaa and assume that $P_1<\infty$.  We denote $f^*$ the unique solution.  Let us consider the SVM defined under $\PP$ by
\beaa da_t= \left( b(a_t)+\sigma(a_t)^2 \partial_a \ln \EE^{\PP^0}[  e^{-f^*(V_T)} | a_t,V_t] \right)dt + \sigma(a_t) dZ_t \eeaa Then $\int_0^T a_s^2 ds \overset{\PP}{\sim} \mu$.
\end{thm} \no The proof is identical to the proof of Proposition \ref{thm1}, and is therefore not reproduced. By using the Bismut-Elworthy-Li formula, the drift term can be put into a Brownian-bridge like-form as in Corollary \ref{corbaudoin}:
\begin{lem}[Bismut-Elworthy-Li formula]
\beaa \partial_a  \EE^{\PP^0}[  e^{-f(V_T)} | a_t,V_t]= {1 \over T-t}\EE^{\PP^0}[e^{-f(V_T)} \int_t^T {Y_s \over \sigma(a_s)} dZ_s^0| a_t,V_t]  \eeaa where the tangent process $Y_s$ is
$ {dY_s \over Y_s}= b'(a_s)ds+\sigma'(a_s) dZ_s^0$.
\end{lem}

\no By differentiating with respect to $f$ the functional $-\EE^\mu[f]-\ln \EE^{\PP^0}[  e^{-f(V_T)} ]$, the optimal potential $f^*$ satisfies
\bea Z \mu(v)= e^{-f^*(v)} \PP^0(V_T=v)  \label{eqfvariance} \eea with $Z$ an irrelevant function. Theorem \ref{thm3} can be simplified into
\begin{cor}[Conditioned SDE on quadratic variation, see also Theorem 4.6 in \cite{bau}] Let us assume that $f^*$ as defined by (\ref{eqfvariance})  is $\mu$-integrable. Let us consider the SVM defined under $\PP$ by
\beaa da_t= \left( b(a_t)+ {\sigma(a_t)^2 \over T-t}{ \EE^{\PP^0}[{\mu(V_T) \over \PP^0(V_T)} \int_t^T {Y_s \over \sigma(a_s)} dZ_s^0| a_t,V_t] \over \EE^{\PP^0}[  {\mu(V_T) \over \PP^0(V_T)} | a_t,V_t]} \right)dt  + \sigma(a_t) dZ_t \eeaa
Then $\int_0^T a_s^2 ds \overset{\PP}{\sim} \mu$.
\end{cor}

\begin{exa}[Entropic SABR model coincides with the Bergomi model \cite{ber0}]
\no Let us consider a SABR model defined under a measure $\PP^0$ by
\beaa
{da_t \over a_t}&=& \nu dZ^0_t \eeaa
We want to modify this model such that we match a $T$-variance swap with payoff $\int_0^T a_s^2 ds$.  Following our previous construction with $f$ restricted to be of the form $f(v)=\lambda v$ with $\lambda \in \RR$, our entropic SVM is
\beaa da_t= \nu^2 a_t^2 \partial_a \ln \EE^{\PP^0}[  e^{-\lambda^* V_T} | a_t,V_t] dt + \nu a_t dZ_t \eeaa where  $\lambda^*$ optimizes
\beaa  \sup_{\lambda  \in \RR} -\lambda \EE^\mu[V_T] -\ln \EE^{\PP^0}[  e^{-\lambda V_T} ] \eeaa Then $\EE^\PP[\int_0^T a^2_sds]= \EE^\mu[V_T]$.  Note that the term
$\partial_a \ln \EE^{\PP^0}[  e^{-\lambda V_T} | a_t,V_t]$  depends only of $a_t$ (not on $V_t$) and from (\cite{bri} -- Dothan's model), we have
\bea \EE^{\PP^0}[  e^{-\lambda (V_T-V_t)} | a_t]={\bar{r}^p \over \pi^2}
\int_0^\infty \sin(2 \sqrt{\bar r}\sinh y)\int_0^\infty f(z) \sin(yz) dzdy+{2 \over \Gamma(2p)}\bar{r}^pK_{2p}(2\sqrt{\bar r}) \nonumber \\ \label{dothan} \eea where
\beaa f(z):=\exp[ -{4 \nu^2(4 p^2+z^2)(T-t) \over 8}] z |\Gamma (-p+{i}{z \over 2})|^2 \cosh {\pi z \over 2 }, \\
\bar{r}={2 \lambda^* a^2_t \over 4 \nu^2}, \quad p={1 \over 2}-\nu^2 \eeaa

\no Now if we impose to be calibrated to all variance swaps for all $T \in \RR^+$ (i.e., $\EE^\PP[a_T^2]=
\EE^{\mu_T}[a^2_T])$, the entropic SABR model becomes
\beaa da_t= \nu^2 a_t^2 \partial_a \ln \EE^{\PP^0}[  e^{-\int_t^{\infty} \lambda^*(s) a^2_s ds} | a_t] dt + \nu a_t dZ_t \eeaa  where
$\lambda^*(\cdot)$ is the solution of:
\beaa \sup_{\lambda(\cdot)} \{ -\int_0^\infty  \lambda(s) \EE^{\mu_s}[a^2_s] ds -\ln \EE^{\PP^0}[  e^{-\int_0^\infty  \lambda(s)  a_s^2 ds} ] \} \eeaa
 By assuming that the term-structure of the variance swaps $ s \in \RR_+ \mapsto \mu_s$ is such that $\lambda^*(s)=\lambda^*$ is constant, we obtain that the dynamics is
 \beaa da_t&=& \nu^2 a_t^2 \partial_a \ln \EE^{\PP^0}[  e^{-\lambda^* \int_t^{\infty}  a^2_s ds} | a_t] dt + \nu a_t dZ_t \\
 &=& \nu^2 a_t^2 \partial_a \ln \EE[e^{-\lambda^* a^2_t \gamma}  ] dt + \nu a_t dZ_t    \eeaa with   $\gamma:=\int_0^{\infty}  e^{-\nu^2 s + 2\nu Z_s^0 }ds$. By taking (in principle) the limit $T \rightarrow \infty$ in (\ref{dothan}),  we have that  the random variable $\gamma$ is distributed according to the inverse of a gamma distribution
 \beaa \gamma  \overset{\PP^0}{\sim} {1 \over 2 G({1 \over 2},1)}, \quad G({1 \over 2},1)(dx):=x^{- {1 \over 2}}{e^{-x} \over \sqrt{\pi} }dx \eeaa
 This is  the so-called Matsumoto-Yor formula \cite{mat}. Finally, by integrating over the  gamma distribution, we obtain
 \beaa {da_t \over a_t}=\nu  dZ^0_t -\sqrt{2 \lambda^*} \nu^2 a_t dt \eeaa  which coincides with the Bergomi model \cite{ber0} without a mean-reversion as the variance swap term-structure has been chosen such that \beaa \EE^{\mu_s}[a_s^2]={ \EE^{\PP^0}[ a_s^2 e^{-\lambda^* V_s} e^{-\sqrt{2 \lambda^*} a_s}  ] \over \EE^{\PP^0}[ e^{-\lambda^* V_s} e^{-\sqrt{2 \lambda^*} a_s}  ] } \eeaa
\end{exa}

\subsection{SVM matching options on the running maximum}

 \noindent Here, we want to add a drift term such that we calibrate (at $t=0$)  market prices of call options on the running maximum $M_T:=\max_{s \in [0, T]} S_s$ (i.e., Lookback options), meaning that we impose $ M_T \overset{\PP}{\sim}\mu$.

\begin{thm}\label{thm3bis} Let us consider the strictly concave optimization problem:
\beaa P_1&:=&\sup_{f \in \mathrm{L}^1(\mu), \Delta_\cdot \in {\cal A}} -\EE^\mu[f] -\ln \EE^{\PP^0}[  e^{-f(M_T)-\int_0^T \Delta_s dS_s} ]\eeaa and assume that $P_1<\infty$.  We denote
 $f^*$ and $\Delta_\cdot^*$ the unique solutions. Let us consider the SVM defined under $\PP$ by
\beaa dS_t&=&S_t a_t dW_t\\
da_t&=& \left( b(a_t)+\sigma(a_t)^2 \partial_a \EE^{\PP^0}[  e^{-f^*(M_T)-\int_t^T \Delta^*_s dS_s} |S_t,M_t,a_t] \right) dt + \sigma(a_t) dZ_t \eeaa  Then $M_T \overset{\PP}{\sim} \mu$.
\end{thm}

\subsubsection{Simplification and HJB equation}
\noindent We have
\beaa \inf_{\Delta_\cdot \in {\cal A}}  \EE^{\PP^0}[  e^{-f(M_T)-\int_0^T \Delta_s dS_s} ]:=u(0,S_0,M_0:=S_0,\pi_0:=0) \eeaa \no where $u(t,s,m,a,\pi)$ is the solution of the HJB:
\beaa \partial_t u + {\cal L}^0u +\inf_{\Delta \in \RR} \left( {1 \over 2} a^2 s^2\partial_\pi^2 u \Delta^2+ \Delta  a^2 s^2 \partial_{s \pi} u+ \rho \Delta \sigma(a) a s \partial_{a \pi} u \right) =0, \quad u(T,s,m,a,\pi)=e^{-f(m)-\pi} \eeaa  with the Neumann condition $\partial_m u(t,m,m,a,\pi)=0$. The solution is $u(t,s,m,a,\pi)=e^{-\pi} U(t,s,m,a,\pi)$ where
\beaa \partial_t U + {\cal L}^0U -{1 \over 2} { \left( a s \partial_{s} U+ \rho \sigma(a)  \partial_{a }U \right)^2 \over U}  =0, \label{BurgersM} \quad U(T,s,m,a)=e^{-f(m)} \eeaa

\subsection{SVM matching Vanillas and VIX options}

\no VIX futures and VIX options, traded on the CBOE, have become popular volatility derivatives. The payoff of a VIX index
at  a future expiry $t_1$ is by definition the price at $t_1$ of the $30$ day log-contract which pays $-{ 2 \over t_2-t_1} \ln {S_{2} \over S_{1}}$ at $t_2=t_1+30$ days: \bea
\mathrm{VIX}_{t_1}^2&\equiv& -{2 \over \Delta } \EE^{\PP^\mathrm{mkt}}_{t_1}\left[\ln
\left(S_{2} \over S_{1}\right)\right] \label{defvix}, \quad \Delta=t_2-t_1 \eea
This definition is at first sight strange as $\mathrm{VIX}_{t_1}$  seems to depend on the probability measure $\PP^\mathrm{mkt}$ (i.e., pricing model) used to value the log-contract at $t_1$. A choice should therefore be made and the probability measure $\PP^\mathrm{mkt}$ selected should be included in the term sheet which describes the payoff to the client. In fact, this conclusion is not correct and the value $\mathrm{VIX}_{t_1}$ is independent of the choice of $\PP^\mathrm{mkt}$ (i.e., model-independence) as it can be replicated at $t_1$ with $t_2$-Vanillas.
\no The payoff of a call option on VIX expiring at $t_1$ with strike $K$ is $\left(\mathrm{VIX}_{t_1}-K\right)^+$.   We want to construct  a SVM calibrated to $t_1$ and $t_2$-Vanillas but also to call options on VIX expiring at $t_1$. We denote $\mu_{\mathrm{VIX}}$ the marginal distribution of  $\mathrm{VIX}_{t_1}^2$ implied from the market and set $C_\mathrm{VIX}(K):=\EE^{\mu_{\mathrm{VIX}}}[\left(\mathrm{VIX}^2_{t_1}-K\right)^+]$ for all $K \in \RR^+$.

\begin{thm}[Vanillas and VIX option]\label{thmvix} Let us consider the strictly concave optimization problem:
\beaa P_3&:=&\sup_{ f_i \in \mathrm{L}^1(\mu_i), (\Delta_s)_{s \in (0,t_2)} \in {\cal A}, \Delta_\mathrm{VIX} \in C^0(\RR_+^3),  f_{\mathrm{VIX}} }\{-\sum_{i=1}^2
\EE^{ \mu_i}[f_i] - \EE^{\mu_{\mathrm{VIX}}}[ f_{\mathrm{VIX}}] \\
&&-\ln \EE^{\PP^0}[ e^{-\sum_{i=1}^2  f_i(S_{t_i})- f_\mathrm{VIX}(X^2)
-\int_{0}^{t_2} \Delta_s dS_s-\Delta_\mathrm{VIX}(S_{t_1},a_{t_1},X)\left( {2 \over t_2-t_1}\ln {S_{t_2} \over S_{t_1}}+X^2 \right)
   }
]   \} \eeaa and assume that $P_3<\infty$. We denote $f_1^*,f_2^*,f_\mathrm{Vix}^*$ and $\Delta_\cdot^*, \Delta_\mathrm{VIX}^*$  the unique solutions. Let us consider the SVM defined under $\PP$ by
\beaa dS_t&=&S_t a_t dW_t \\
da_t&=&\sigma(a_t)^2 \partial_a \EE_t^{\PP^0}[  e^{-  f_2(S_{t_2})-f_1(S_{t_1})1_{t <t_1}- f_\mathrm{VIX}(X^2)
-\int_{t}^{t_2} \Delta_s dS_s-\Delta_\mathrm{VIX}(S_{t_1},a_{t_1})\left( {2 \over t_2-t_1}\ln {S_{t_2} \over S_{t_1}}+X^2 \right)
} ] dt \\&+& b(a_t)dt+\sigma(a_t) dZ_t \eeaa Then $S_{t_i} \overset{\PP}{\sim} \mu_i,\quad i=1,2$ and $
\mathrm{VIX}_{t_1}^2 \overset{\PP}{\sim} \mu_\mathrm{VIX}$.
\end{thm}

\section{Link with   Dyson Brownian motions and random matrices}

\no The ordered eigenvalues $X^1_t \leq X^2_t  \leq \cdots \leq  X^n_t$
of a Brownian motion in the space of $n \times n$ real Hermitian matrices form a diffusion process which satisfies the SDE:
\beaa dX_t^i=\sum_{1 \leq j \leq n\;:\; j\neq i} {dt \over X_t^i-X_t^j} + dW_t^i \eeaa
where $(W_t^i)_{1 \leq i\leq n}$ are $n$ independent real Brownian motions. This  result  goes back to Dyson
\cite{dys} and corresponds to non-colliding Brownian motions.  The solution of the associated Fokker-Planck equation can be explicitly solved as done by  Johansson in  \cite{joh}:
\beaa p_\mathrm{Dyson}(t,x|y)={1 \over (2 \pi t)^{n \over 2}} {\Delta_n(x) \over \Delta_n(y)} \det [ e^{-{(x_i-y_j)^2 \over 2 t}}]_{1 \leq i,j \leq n} \eeaa with $\Delta_n(x):=\prod_{i<j}^n |x_i-x_j|$.   We would like to reproduce this result by interpreting the Dyson SDE as a Schr\"{o}dinger bridge and then plans to obtain the joint probability density $p_\mathrm{Dyson}$ using the Schr\"{o}dinger factorization property. In this purpose, we consider the optimization problem:
\beaa P_\mathrm{nc}:=\inf_{\PP \in {\cal M}_{nc} } H(\PP|\PP^0)  \eeaa over the convex space ${\cal M}_\mathrm{nc}$ of non-colliding measures \beaa
{\cal M}_\mathrm{nc}:=\{ \PP \sim \PP^0 \; : \; X^1_t \leq X^2_t  \leq \cdots \leq  X^n_t, \quad \PP-\mathrm{as}. \} \eeaa where $\PP^0$ is the $n$-dimensional Wiener measure. For use below, we introduce the  function ${\cal H}(t,x)=1-1_{x^n \geq x^{n-1}  \geq \cdots \geq x^1}$. Note that for all $\PP \in {\cal M}_\mathrm{nc}$, we have $\EE^\PP[{\cal H}(t,X_t)]=0$ where the subscript ``$ \mathrm{nc}$'' means ``non-colliding''.  From the previous section, ${\cal H}$ can be interpreted as a payoff for $n$ underlyings with zero market price.  We have:

\begin{thm} Let us consider the SDE under a measure $\PP$:
\beaa dX_t=\nabla_x \ln \EE^{\PP^0}[ e^{-\int_t^\infty ds \lambda(s) {\cal H}(s,W^0_s)}|X_t]+dW_t \eeaa where $\lambda(\cdot)$ is the solution of the concave optimization:
\beaa \sup_{\lambda(\cdot)} -\ln \EE^{\PP^0}[ e^{-\int_0^\infty ds \lambda(s) {\cal H}(s,W^0_s)}|X_0] \eeaa Then $\PP \in {\cal M}_\mathrm{nc}$.
\end{thm}
\no The supremum is attained for $\lambda^*(\cdot)=\infty$ as the functional $-\ln \EE^{\PP^0}[ e^{-\int_0^\infty ds \lambda(s) {\cal H}(s,W^0_s)}|X_0] $ is increasing in $\lambda(\cdot)$ and we get the SDE
\beaa dX_t=\nabla_x \ln \EE^{\PP^0}[  1_{\cal W} |X_t]+dW_t \eeaa where  ${\cal W}:=\{X^1_t \leq X^2_t  \leq \cdots \leq  X^n_t, \quad \forall t \}$.
By computing  $\EE^{\PP^0}[1_{\cal W} |X_t]$ in closed-form using the Karlin-McGregor formula \cite{karlin}, we reproduce the Dyson SDE and this implies:
\begin{cor} The unique solution of $P_\mathrm{nc}$ is the Dyson SDE.
\end{cor} \no This result is in line with the well-known statement that the Gaussian distribution is the minimal entropy density with fixed mean and variance (Maxwellian distribution). The Dyson Brownian motion is also the  minimal entropy diffusion if we restrict the particles to be non-colliding.
\no This implies the following density factorization of the joint probability density:
\begin{cor} $p_\mathrm{Dyson}$ can be factorized as
\beaa p_\mathrm{Dyson}(t,x)=\Psi(t,x) \bar{\Psi}(t,x) \eeaa
\end{cor}

\begin{exa}[Brownian excursion and meander]  We consider here a Brownian motion  starting at $X_0:=0$ which is constrained to stay positive between $[0,T]$ and such that $X_T \overset{\PP}{\sim} \mu$. This can be seen as a variant of a Brownian excursion and meander. A Brownian excursion  is a  Brownian motion (i.e., $X_0:=X_T:=0$) which is constrained to stay positive between $[0,T]$. A Brownian meander is a Brownian motion which starts at $X_0:=0$ and which is constrained to stay positive between $[0,T]$. Following the same construction as above, we obtain that
\beaa dX_t=\nabla_x \ln \EE^{\PP^0}[  1_{\cal A} e^{-f^*(X_T)}|X_t]dt+dW_t \eeaa where  ${\cal A}:=\{X_t \geq 0, \quad \forall t \in [0,T] \}$ and $f^*$ solution of
\beaa \sup_{f \in \L^1(\mu)} -\EE^{\mu}[f]-\ln \EE^{\PP^0}[  1_{\cal A} e^{-f(X_T)}] \eeaa Using the method of images, this can be simplified into
\beaa dX_t=\nabla_x \ln \int \mu(dy) {p_{\cal A}(T,y|t,X_t) \over p_{\cal A}(T,y|0,X_0:=0)}+dW_t \eeaa where
$p_{\cal A}(T,y|t,x)={1 \over \sqrt{2 \pi(T-t)}}\left( e^{-(y-x)^2 \over 2(T-t)}-e^{-{(y+x)^2 \over 2(T-t)}} \right)$.
\end{exa}

\appendix
\section*{Some proofs}

\begin{proof}[Proof of Proposition \ref{thm1}]
\no Let us consider the following stochastic control problem:
\bea P_1:=\inf_{\PP \in {\cal M}(\mu)} H(\PP|\PP^0)  \label{primal} \eea
where ${\cal M}(\mu):=\{ \PP \sim  \PP^0 \;: \; X_{T} \overset{\PP}{\sim} \mu \}$ and $H(\PP|\PP^0):=\EE^{\PP}[ \ln {d\PP \over d\PP^0} ]$ is the relative entropy with respect to a prior $\PP^0$ chosen to be the Brownian measure. From the Girsanov theorem, for all $\PP \sim \PP^0$, we have the following dynamics under $\PP$:
\beaa dX_t=\lambda_t dt + dW_t \eeaa where $W_t$ is a $\PP$-Brownian motion and
\beaa {d\PP^0 \over d\PP}|_{{\cal F}_T}=e^{-\int_0^{T} \lambda_s dW_s -{1 \over 2 } \int_0^{T} \lambda_s^2 ds} \eeaa
\no Under $\PP^0$, $X_t$ is a $\PP^0$-Brownian motion. We have therefore
 \beaa H(\PP|\PP^0)&=&\EE^{\PP}[
\left( \int_0^{T} \lambda_s dW_s +{1 \over 2 } \int_0^{T} \lambda_s^2 ds \right)]={1 \over 2 }\EE^{\PP}[ \int_0^{T} \lambda_s^2 ds]  \eeaa  By convex duality, the primal problem (\ref{primal}) can be converted into the unconstrained optimization:
\beaa
P_1=\sup_{f \in \mathrm{L}^1(\mu)} \inf_{\lambda_\cdot \in {\cal A} } {1 \over 2 }\EE^{\PP}[ \int_0^{T} \lambda_s^2 ds+f(X_{T})]-\EE^{\mu}[f]  \eeaa
This is equivalent to
\beaa P_1=\sup_{f \in \mathrm{L}^1(\mu)}\{ u(0,X_0)-\EE^{\mu}[f]  \} \eeaa where $u(t,x):=\inf_{\lambda_\cdot \in {\cal A} } {1 \over 2 }\EE^{\PP}[ \int_t^{T} \lambda_s^2 ds+f(X_{T})|X_t=x]$ is the solution of the Hamilton-Jacobi-Bellman PDE:
\beaa \partial_t u(t,x) + {1 \over 2}\partial_x^2 u  +\inf_{\lambda} \{ {1 \over 2 }\lambda^2
+\lambda \partial_x u\}=0, \quad u(T,x):=f(x) \eeaa
\no Taking the infimum over $\lambda$, we get the Burgers PDE:
\beaa \partial_t u+ {1 \over 2}\partial_x^2 u  - {1 \over 2} (\partial_x u)^2=0, \quad u(T,x)=f(x)  \eeaa
where the optimal control is $\lambda^*_t=-\partial_x u(t,X_t)$. By using a Cole-Hopf transformation $u:=-\ln U$, this PDE can be transformed into the heat kernel:
\beaa \partial_t U + {1 \over 2}\partial_x^2 U=0, \quad U(T,x)=e^{-f(x)}  \eeaa for which we deduce the solution:
\beaa U(t,x)=\EE[ e^{-f(W_T^0)}|W_t^0=x] \eeaa   Finally, our primal reads
\beaa P_1:=\sup_{f \in L(\mu)}\{ -\EE^{\mu}[f] -\ln \EE[ e^{-f(W_T^0)}|W_0^0=X_0]  \} \eeaa and this concludes the proof with our expression of the optimal control $\lambda^*_t$.  This (modern) version of the proof of the Schr\"{o}dinger result is due to F\"{o}llmer \cite{fol}.
\end{proof}

\begin{proof}[Proof of Proposition \ref{thm1bis}]\no Let us consider the following stochastic control problem:
\bea P_2:=\inf_{\PP \in {\cal M}(\mu_1,\mu_2)} H(\PP|\PP^0)  \label{primalbis} \eea By convex duality, $P_2$ can be converted into
\beaa
P_2=\sup_{f_1 \in \mathrm{L}^1(\mu_1),f_2 \in \mathrm{L}^1(\mu_2)} \inf_{\lambda_\cdot  \in {\cal A} } {1 \over 2 }\EE^{\PP}[ \int_0^T \lambda_s^2 ds+f_1(X_{t_1})+f_2(X_{t_2})]-\EE^{\mu_1}[f_1] -\EE^{\mu_2}[f_2] \eeaa The proof is then identical to the proof of Proposition \ref{thm1}.
\end{proof}

\begin{proof}[Proof of Proposition \ref{thm1bisbis}] With the entropic construction, the optimal drift is a function of the time $t$ and $X_t$: $dX_t=\lambda^*(t,X_t) dt +dW_t$. Below, we will show that this drift is completely fixed if we prescribe the $t$-marginals of $X_t$ to be $\mu_t$ for all $t>0$.  By applying It\^o-Tanaka on the convex payoff $(X_t-K)^+$, we get
\beaa d(X_t-K)^+=1_{X_t>K} dX_t +{1 \over 2} \delta(X_t-K) dt \eeaa where $dL_t^K:={1 \over 2} \delta(X_t-K) dt$ is interpreted as the local time at $K$. By taking the expectation on both sides, this gives
\beaa \partial_t  \EE[(X_t-K)^+]=\EE[1_{X_t>K} \lambda(t,X_t)] +{1 \over 2} \EE[\delta(X_t-K)] \eeaa By differentiating w.r.t. $K$, we obtain
\beaa \partial_t  F(t,K)&=&-\lambda(t,K) \partial_K F(t,K) +{1 \over 2} \partial^2_K  F(t,K) \eeaa
 where $F(t,K):=\EE[1_{X_t<K}]$. Using our formula for $\lambda$ in Proposition \ref{thm1bisbis}, we have
\beaa \partial_t  F(t,K)&=&\left( {\partial_t  F_\mu(t,K) -{1 \over 2}\partial_K^2 F_\mu(t,K) \over \partial_K F_\mu(t,K)}\right)\partial_K F(t,K)+{1 \over 2} \partial^2_K  F(t,K)  \eeaa We conclude by the uniqueness of this linear PDE that the solution is $F=F_\mu$ and therefore
$X_t \overset{\PP}{\sim} \mu_t$ for all $t\in \RR_+$.
\end{proof}

\begin{proof}[Proof of Theorem \ref{thm2}]

\no \no Let us consider the following stochastic control problem:
\beaa P_1:=\inf_{\PP \in {\cal M}_\mathrm{mart}(\mu_1)} H(\PP|\PP^0)  \label{primaln} \eeaa
where ${\cal M}_\mathrm{mart}(\mu_1):=\{ \PP \sim  \PP^0 \;: \; S_t \; \PP-\mathrm{martingale}, \quad S_{t_1} \overset{\PP}{\sim} \mu_1 \}$.  For all  $\PP \sim \PP^0$ (not necessary martingale measure here), we have
\beaa dS_t&=&S_t a_t (dW_t+\lambda_t^1 dt)\\
da_t&=&b(a_t) dt + \sigma(a_t) \left( \rho  (dW_t+\lambda_t^1 dt)+  \sqrt{1-\rho^2} (dW^\perp_t+\lambda_t^2 dt) \right),  \quad d\langle W, W^\perp \rangle_t=0\eeaa where $\lambda^{1,2}_t$ are arbitrary adapted processes in $\cal A$. For convenience, we have decomposed $Z$ and $W$ into two uncorrelated Brownian motions $W$ and $W^\perp$. From the Girsanov theorem, we have
\beaa {d\PP^0 \over d\PP}|_{{\cal F}_t}=e^{-\int_0^t \lambda_s^{(1)} dW_s-{1 \over 2} \int_0^t
(\lambda_s^{(1)})^2 ds}e^{-\int_0^t \lambda_s^{(2)} dW^\perp_s-{1 \over 2} \int_0^t
(\lambda_s^{(2)})^2 ds} \eeaa This implies that
\beaa H(\PP|\PP^0) = {1 \over 2} \sum_{i=1}^2 \EE^\PP[\int_0^{{t_1}}
(\lambda_s^{(i)})^2]ds \eeaa Following closely the proof of Proposition \ref{thm1}, we can dualize  $P_1$ into
\beaa P_1&=&\sup_{f_1 \in \mathrm{L}^1(\mu_1), \Delta_s \in {\cal A} }  \{ -\EE^{ \mu_1}[f_1]+\inf_{\lambda_s \in {\cal A}}
{1 \over 2} \sum_{i=1}^2 \EE^\PP[\int_0^{{t_1}}
(\lambda_s^{(i)})^2ds +f_1(S_{t_1})+\int_0^{t_1} \Delta_s dS_s]   \}  \eeaa
Note that the martingale condition has been imposed by introducing the It\^o integral $\int_0^{t_1}  \Delta_s dS_s$ with $\Delta  \in {\cal A}$, the space of adapted process in $\mathrm{L}^2(\PP)$, for which
\beaa \EE^{\PP}[ \int_0^{t_1}  \Delta_s dS_s]=0 \eeaa if and only if $S_t$ is a $\PP$-martingale. Finally, $P_1$ can be written as
\beaa P_1&=&\sup_{f_1 \in \mathrm{L}^1(\mu_1)}\{ -\EE^{ \mu_1}[f_1]+u(0,S_0,a_0)  \} \eeaa
 where $u(t,s,a):=\inf_{\lambda_s \in {\cal A}}
{1 \over 2} \sum_{i=1}^2 \EE^\PP[\int_t^{{t_1}}
(\lambda_s^{(i)})^2ds +f_1(S_{t_1})+\int_t^{t_1} \Delta_s dS_s|S_t=s,a_t=a]$ is the solution of the HJB PDE:
\beaa \partial_t u+ {\cal L}^0 u  +\inf_{\lambda_1,\lambda_2 \in \RR^2} \sup_{\Delta \in \RR} \{ \sum_{i=1}^2 {(\lambda^{(i)})^2 \over 2}
+\lambda^1 (s a \partial_s u+\rho \sigma(a) \partial_a u+sa \Delta)
+\lambda^2 \sqrt{1-\rho^2} \sigma(a) \partial_a u
\}=0 \\ \quad u(t_1,s,a)=f_1(s) \eeaa
and where ${\cal L}^0:={1 \over 2}s^2 a^2 \partial_{ss}+{1 \over 2}\sigma(a)^2 \partial_{aa} + \rho \sigma(a) a s \partial_{as}+b(a)\partial_a $ is the It\^o generator of the process $(S_t,a_t)$ under $\PP^0$. By taking the infimum over $\lambda^{(i)}$, we get the Burgers-like equation:
\beaa \partial_t u + {\cal L}^0u  - {1 \over 2 }
\left( (\sigma(a) \partial_a u)^2+(s a \partial_s u)^2 +2 \rho \sigma(a) s a \partial_{s}u \partial_a u\right) \\+
\sup_{\Delta \in \RR} \{ -{(sa \Delta)^2 \over 2}-s a \Delta
\left(s a  \partial_su +\rho \sigma(a) \partial_a u \right) \} =0 \eeaa
and the optimal controls are \beaa
(\lambda^*)^{(1)}_t&=& -\left(S_t a_t \partial_s +\rho \sigma(a_t) \partial_a \right)u(t,S_t,a_t)-S_t a_t \Delta_t\\
(\lambda^*)^{(2)}_t&=&-\sqrt{1-\rho^2}\sigma(a_t)\partial_a u(t,S_t,a_t) \eeaa
\no Note that if we set $u:=-\ln U$, then
\beaa \partial_t U + {\cal L}^0U +
\sup_{\Delta \in \RR}  \{ {(sa \Delta)^2 \over 2}U-s a \Delta
\left(s a  \partial_sU +\rho \sigma(a) \partial_a U \right) \} =0 \eeaa
\no By taking the supremum over $\Delta$, on gets \beaa
\Delta_t^*=-\left( \partial_s +\rho {\sigma(a_t) \over S_t a_t} \partial_a \right)u \eeaa and
$ (\lambda_t^*)^{(1)}=0 $ as expected as $S_t$ should be a martingale.  Then, $u$ is solution of the Burgers-like PDE:
\bea \partial_t u + {\cal L}^0u  - {1 \over 2 }(1-\rho^2) (\sigma(a) \partial_a u)^2=0, \quad u(t_1,s,a)=f_1(s) \label{Burgersa} \eea
\no The solution $u$ of (\ref{Burgersa}) could then be written back as a stochastic control problem:
\beaa u(t,s,a)=\sup_{\Delta_s \in {\cal A}} - \ln \EE^{\PP^0}[ e^{-f_1(S_{t_1}) -\int_t^{t_1} \Delta_s dS_s} | S_t=s,a_t=a] \eeaa
Note that this implies that the Radon-Nikodym derivative ${d\PP \over d\PP^0}|_{{\cal F}_t}$ is
\beaa {d\PP \over d\PP^0}|_{{\cal F}_t}&=&{e^{-u(t,S_t,a_t)-\int_0^t \Delta_s^* dS_s} \over e^{-u(0,S_0,a_0)}} \\
&=&{\EE^{\PP^0}[ e^{-f_1(S_{t_1}) -\int_0^{t_1} \Delta^*_s dS_s} | S_t,a_t] \over \EE^{\PP^0}[ e^{-f_1(S_{t_1}) -\int_0^{t_1} \Delta_s dS_s}]}e^{-\int_0^t \Delta_s^* dS_s}  \eeaa
\no $P_1$ can then be written as
\beaa P_1=\sup_{f_1 \in \mathrm{L}^1(\mu_1), \Delta_s \in {\cal A}}\{ -\EE^{ \mu_1}[f_1]
-\ln \EE^{\PP^0}[ e^{-f_1(S_{t_1})-\int_0^ {t_1} \Delta_s dS_s} ]   \} \eeaa  and
\beaa dS_t&=&S_t a_t dW_t\\
da_t&=&\left( b(a_t) +(1-\rho^2)\sigma(a_t)^2\partial_a  \ln \EE^{\PP^0}[ e^{-f_1(S_{t_1}) -\int_t^{t_1} \Delta^*_s dS_s} | S_t,a_t] \right)dt + \sigma(a_t) dZ_t \eeaa
\end{proof}

\begin{proof}[Proof of Theorem \ref{thm2bis}]
\no

\no  Let us consider the following stochastic control problem:
\bea P_2:=\inf_{\PP \in {\cal M}_\mathrm{mart}(\mu_1,\mu_2) } H(\PP|\PP^0)  \label{primaln2} \eea
where ${\cal M}_\mathrm{mart}(\mu_1,\mu_2):=\{ \PP \sim  \PP^0 \;: \; \; S_t \; \PP-\mathrm{martingale}, \quad S_{t_1} \sim \mu_1, S_{t_2} \sim \mu_2\}$.  Following closely the proof of Theorem \ref{thm2}, we obtain that $P_2$ can then be written as
\beaa P_2=\sup_{ f_1 \in \mathrm{L}^1(\mu_1), f_2 \in \mathrm{L}^1(\mu_2), (\Delta_s)_{s \in (0,t_2)} \in {\cal A}}\{-\EE^{ \mu_1}[f_1] -\EE^{ \mu_2}[f_2]
-\ln \EE^{\PP^0}[ e^{-f_1(S_{t_1})-f_2(S_{t_2})-\int_{0}^{t_2} \Delta_s dS_s} ]   \} \eeaa
\end{proof}

\begin{proof}[Proof of Theorem \ref{thmvix}]
\no

\no  Let us consider the following stochastic control problem:
\bea P_3:=\inf_{\PP \in {\cal M}_\mathrm{mart}(\mu_1,\mu_2,\mu_\mathrm{VIX}) } H(\PP|\PP^0)  \eea
where ${\cal M}_\mathrm{mart}(\mu_1,\mu_2):=\{ \PP \sim  \PP^0 \;: \; \; S_t \; \PP-\mathrm{martingale}, \quad S_{t_1} \sim \mu_1,\quad S_{t_2} \sim \mu_2, \quad \mathrm{VIX}_{t_1}^2 \sim \mu_\mathrm{VIX}
\}$.  The program $P_3$ can be dualized into
\bea P_3&:=&
\sup_{ f_i \in \mathrm{L}^1(\mu_i), (\Delta_s)_{s \in (0,t_2)} \in {\cal A}, \Delta_\mathrm{VIX} \in C^0(\RR_+^3),  f_{\mathrm{VIX}} \in
\mathrm{L}^1(\mu_\mathrm{VIX})}
\inf_{\PP \sim \PP^0} H(\PP|\PP^0) \\
&+&\EE^\PP[\int_0^{t_2} \Delta_s dS_s]+\sum_{i=1}^2 \EE^{\PP}[ f_i] - \EE^\mu[f_i]\\
&+&\EE^{\PP}[ \Delta_\mathrm{VIX}(S_{t_1},a_{t_1},X)\left( {2 \over t_2-t_1}\ln {S_{t_2} \over S_{t_1}}+X^2 \right)]
\\
&+&\EE^{\PP}[ f_{\mathrm{VIX}}(X^2)] -\EE^{\mu_{\mathrm{VIX}}}[ f_{\mathrm{VIX}}]
 \eea

\no Following closely the proof of Theorem \ref{thm2}, we obtain that $P_3$ can then be written as
\beaa P_3&=&\sup_{ f_i \in \mathrm{L}^1(\mu_i), (\Delta_s)_{s \in (0,t_2)}, \Delta_\mathrm{VIX} \in C^0(\RR_+^2),  f_{\mathrm{VIX}} }\{-\sum_{i=1}^2
\EE^{ \mu_i}[f_i] - \EE^{\mu_{\mathrm{VIX}}}[ f_{\mathrm{VIX}}] \\
&&-\ln \EE^{\PP^0}[ e^{-\sum_{i=1}^2  f_i(S_{t_i})- f_\mathrm{VIX}(X^2)
-\int_{0}^{t_2} \Delta_s dS_s-\Delta_\mathrm{VIX}(S_{t_1},a_{t_1},X)\left( {2 \over t_2-t_1}\ln {S_{t_2} \over S_{t_1}}+X^2 \right)
   }
]   \} \eeaa
\end{proof}

\end{document}